\documentclass[12pt,draftcls,peerreview,onecolumn]{IEEEtran}

\ifCLASSINFOpdf
  \usepackage[pdftex]{graphicx}
\else
\fi

\usepackage{graphicx}
\usepackage{amsmath}
\usepackage{footnote}
\usepackage{amsthm}
\usepackage{amssymb}
\newtheorem{definition}{Definition}
\newtheorem{theorem}{Theorem}
\newtheorem{proposition}{Proposition}
\newtheorem{algorithm}{Algorithm}
\newtheorem{lemma}{Lemma}
\usepackage{float}
\usepackage{latexsym}
\usepackage{setspace}
\usepackage{subfigure} 
\usepackage[colorlinks,citecolor=blue, linkcolor=blue]{hyperref}
\usepackage{tabularx}
\usepackage{varwidth}
\usepackage{multirow}
\usepackage{color}
\usepackage{calligra}
\usepackage{algorithm}
\usepackage{algorithmicx}
\usepackage{algpseudocode}

\begin{document}
\title{Fast Graph Subset Selection Based on G-optimal Design}

\author{Zhengpin Li$^1$, Zheng Wei$^1$, Jian Wang$^1$, Yun Lin$^2$ and Byonghyo Shim$^{1,3}$ \\
$^1$School of Data Science, Fudan University, China 
\\
$^2$School of Information and Communication Engineering, Harbin Engineering University, China \\
$^3$Department of Electrical and Computer Engineering, Seoul National University,  Korea\\
E-mail: jian\_wang@fudan.edu.cn 
%
}
\date{December 25, 2021}

\IEEEtitleabstractindextext{%
\begin{abstract}
Graph sampling theory extends the traditional sampling theory to graphs with topological structures. As a key part of the graph sampling theory, subset selection chooses nodes on graphs as samples to reconstruct the original signal. Due to the eigen-decomposition operation for Laplacian matrices of graphs,  however, existing subset selection  methods usually require high-complexity calculations. In this paper, with an aim of enhancing the computational efficiency of subset selection on graphs, we propose a novel objective function based on the optimal experimental design. Theoretical analysis shows that this function enjoys an $\alpha$-supermodular property with a provable  lower bound on $\alpha$. The objective function, together with an approximate of the low-pass filter on graphs, suggests a fast subset selection method that does not require any eigen-decomposition operation. Experimental results show that the proposed method exhibits high computational efficiency, while having competitive results compared to the state-of-the-art ones, especially when the sampling rate is low.   
\end{abstract}
\begin{IEEEkeywords}
Graph signal processing, sampling theory, optimal experimental design. 
\end{IEEEkeywords}}
\maketitle

\IEEEdisplaynontitleabstractindextext

%
\IEEEpeerreviewmaketitle

\section{Introduction}
{\bf Notice: This work has been submitted to the IEEE for possible publication. Copyright may be transferred without notice, after which this version may no longer be accessible.}

\IEEEPARstart{I}{n} various scenarios such as sensor networks, social networks, transportation systems and biomedicine, the signals of interest are defined on graphs~\cite{Overviewchallengesandapplications,Theemergingfieldofsignalprocessingongraphs}. Unlike images and time-varying signals, graph signals have a complex and irregular topology~\cite{Bigdataanalysis,DSPsampling}. By exploting potential connection information, graph signal processing (GSP) extends the mature tools for analyzing classical signals to graph signals~\cite{zhu2012approximating}. There are two main topics in GSP, which can be understood as encoding and decoding procedures. 
When faced with large-scale and complex structured graph data in the real world~\cite{Theemergingfieldofsignalprocessingongraphs,DSPsampling}, processing over raw data often consumes lots of storage and calculation resources, which motivates economic encoding paradigm that can effectively capture most of the information. Roughly speaking, sampling of graph signals can be studied under different perspectives: subset selection~\cite{chen2015sampling}, aggregation sampling~\cite{successivelocalaggregations,randomizedlocalaggregations} and local measurement~\cite{wang2016local}. 

In a nutshell, subset selection chooses nodes on the graph as samples to reconstruct graph signals that meet certain conditional assumptions. A typical example is active semi-supervised learning~\cite{Activesemi-supervised}, where the learner selects a small number of points as the training set, gets their labels and predicts the unknown labels of other points. Since the vector of classification tags is naturally a graph signal, selecting the training set can be translated to subset selection on graphs. In contrast to traditional sampling scenarios where Nyquist sampling theorem applies, subset selection on graphs has to consider the topological structure of nodes. Existing methods for subset selection can be roughly divided into two categories: i) random sampling and ii) deterministic sampling.  For random sampling, it has been shown that uniformly selecting a sufficient number of points in Erd{\"o}s-R\'{e}nyi graph can ensure perfect recovery~\cite{DSPsampling}. Besides, nonuniform selection of points according to experimentally designed sampling probabilities has been studied based on the restricted isometry property in compressed sensing~\cite{randomizedlocalaggregations,Randomsampling,10.1093/imaiai/iax021}. For both selection manors in the random category, experimental and theoretical comparisons with respect to the reconstruction error have been reported in~\cite{7148908}. 

Deterministic sampling chooses a fixed point subset according to a pre-determined objective function. Compared to random sampling, which performs well in the average sense, deterministic sampling usually requires fewer sampling points in order to achieve the same reconstruction effect. Among methods in the deterministic family, vertex-based methods have been widely studied in the field of machine learning as a sensor placement problem~\cite{krause2008near}. Recently, with the development of graph spectral theory, spectrum-based methods have also been proposed. They assume i) that the eigendecomposition of graph Laplacian is explicitly known and ii) that graph signals meet certain conditions (e.g., bandlimited or smooth). In the noiseless case, it has been proved that when the sampling budget is larger than the bandwidth, accurate recovery can always be ensured~\cite{DSPsampling}. In the presence of noise, there are deterministic methods connecting subset selection with optimal experiment design (OED), which introduces some scalarized statistical criteria as objective functions that are related to the variance matrix of the estimator. Typical criteria include determinant, trace and norm, which have inspired popular sampling methods such as D-optimal~\cite{Uncertaintyprincipleandsampling}, A-optimal~\cite{Uncertaintyprincipleandsampling} and E-optimal~\cite{DSPsampling}. Although those problems are combinatorial and NP-hard in their original forms, practical solutions can be found through greedy search~\cite{Samplinglarge} or convex relaxations~\cite{8309040}.

Since the spectrum-based methods have to calculate the eigen-decomposition of graph Laplacian in advance, it generally requries high-complexity calculations (i.e., $\mathcal{O}(N^3)$ where $N$ is the number of nodes in a graph)~\cite{tanaka2020sampling}. To alleviate the computational burden, methods without eigen-decomposition have been developed. The underlying strategies include using graph spectral proxies to maximize the approximate bound on the cut-off frequency for a given sampling set~\cite{Efficientsamplingsetselection}, assuming a biased graph Laplacian regularization based scheme~\cite{9040409}, or employing the graph localization operator~\cite{Eigendecomposition-free, Adistance-basedformulation}. While these method greatly improve the calculation efficiency, they often lack performance guarantees.

In this paper, with an aim of providing theoretical guarantees and meanwhile enhancing the computational efficiency,  we propose a subset selection method called augmented G-optimal design (AGOD) for graph subset selection. Our method is built upon a new statistical criterion called G-optimal experimental design that serves as a greedy principle for choosing informative sampling points in the graph. We give a theoretical analysis to characterize the convergence behavior of the greedy search process. The AGOD method, together with an approximate of a low-pass filter with given cutoff frequency, leads to a cost-effective version called fast AGOD (FAGOD) without any eigen-decomposition of Laplacian matrix. The main contributions of our paper are summarized as follows.

%

\begin{enumerate}
	
	\item Using the supermodular theory, we analyze the theoretical performance of AGOD. Our results show that the objective function, which minimizes the worse-case of reconstruction error, has an $\alpha$-supermodular property with $$\alpha \geq \frac{(2+\mu)\mu}{(1+\mu)^2}.$$ This ensures that the subset of points chosen by AGOD is near-optimal. 
	
	\item By reformulating the objective function, we apply fast graph Fourier transform to avoid the eigen-decomposition of Laplacian matrix, which motivates the FAGOD algorithm for selecting sampling points in a greedy fashion. Experimental results show that FAGOD has competitive performance compared to the state-of-the-art methods, especially when the sampling rate is low (e.g., when the sampling budget equals the bandwidth of target signal). 
\end{enumerate}

We stress that similar to the proposed FAGOD algorithm,~\cite{8827302} also suggested a counterpart called graph filter submatrix-based sampling (GFS) that uses fast graph Fourier transform to replace the eigen-decomposition of Laplacian matrix. However, FAGOD has theoretical advantage on the convergence property over GFS~\cite{8827302}. This is mainly because FAGOD is built upon a completely new objective criterion (i.e., G-optimal design) that enjoys $\alpha$-supermodular property with a tighter lower bound of $\alpha$, which, in essence, indicates a faster convergence in  greedy search.


The rest of this paper is organized as follows. In Section~\ref{sec3}, we overview GSP concepts and other basic knowledge used throughout this paper. We propose our new sampling criteria based on G-optimal design and talk about the theoretical properties in Section~\ref{sec44} and then a fast optimization algorithm is detailed in Section~\ref{sec5}. Finally, we present experimental results and conclusion in Sections~\ref{sec6} and~\ref{sec7}, respectively.

\section{Preliminaries}\label{sec3}
In this section, we will first introduce some basic knowledge and definitions of GSP, and then formally describe the problem of sampling on the graph. Finally, we will briefly review the existing sampling methods based on the ODE. Notations used in this paper are summarized in Table~\ref{table1}.

\subsection{Graph Signal Processing}
We briefly introduce signal processing on graphs. Consider a graph \begin{equation}
\mathcal{G}=(\mathcal{V}, \mathbf{A}),
\end{equation} 
where \(\mathcal{V}=\left\{v_{0}, \ldots, v_{N-1}\right\}\) is the set of nodes, and \(\mathbf{A} \in \mathbb{R}^{N \times N}\) is the weighted adjacency matrix, where the edge weight \(\mathbf{A}_{i, j}\) between nodes \(v_{i}\) and \(v_{j}\) represents the underlying relation between them. While the graph is undirected, the adjacency is always symmetry. Given $\mathbf{A}$, the graph Laplacian matrix $\mathbf{L}$ is computed as
\begin{equation}
	\mathbf{L}=\mathbf{D}-\mathbf{A},
\end{equation}
where $\mathbf{D}\in \mathbb{R}^{N \times N}$ denotes the diagonal degree matrix and $\mathbf{D}_{ii} = \sum_j\mathbf{A}_{ij}$. Without loss of generality, suppose that the node index is fixed, and thus the graph signal can be rewritten as a vector form:
\begin{equation}
	\mathbf{x}=\left[x_{0} ,\ldots , x_{N-1}\right]^\top \in \mathbb{R}^{N}.
\end{equation}
The eigen-decomposition of \(\mathbf{L}\) is 
\begin{equation}
	\mathbf{L}=\mathbf{V} \mathbf{\Lambda} \mathbf{V}^{\top},
\end{equation}
where the eigenvectors of $\mathbf{L}$ form the columns of matrix \(\mathbf{V}\in\mathbb{R}^{N \times N}\) whose each column is normalized to have unit $\ell_2$-norm, and the eigenvalue matrix \(\mathbf{\Lambda} \in \mathbb{R}^{N \times N}\) is the block diagonal matrix of the corresponding eigenvalues \(\lambda_{0}, \ldots, \lambda_{N-1}\) of $\mathbf{L}$ in ascending order. These eigenvalues represent frequencies on the graph~\cite{DiscretesignalprocessingongraphsFrequency,Vertex-frequency}.

In order to formulate a sampling theory for graph signals, we need a notion of frequency that helps to characterize the level of bandlimitedness of the graph signal with respect to the graph. The graph Fourier transform (GFT) of a graph signal \(\mathbf{x} \in \mathbb{R}^{N}\) is defined as~\cite{DSPsampling}:
\begin{equation}
	\mathbf{\widehat{{x}}}=\mathbf{V}^{\top} \mathbf{x} \in \mathbb{R}^{N},
\end{equation}
while the inverse graph Fourier transform is
\begin{equation}
\mathbf{x}=\mathbf{V} \mathbf{\widehat{{x}}}=\sum_{k=0}^{N-1} \widehat{x}_{k} \mathbf{v}_{:k}	 ,
\end{equation}
where $\mathbf{v}_{:k} \in \mathbb{R}^{N}$ is the $k$-th column of $\mathbf{V}$ and $ \widehat{x}_{k}$ is the $k$-th component in $ \mathbf{\widehat{x}}$. The vector $\mathbf{\widehat{{x}}}$ represents the signal’s expansion in the eigenvector basis and describes the frequency components of the graph signal $\mathbf{x}$. We mention that there are also other ways to define GFT, such as using the Jordan decomposition of adjacent matrix $\mathbf{A}$~\cite{chen2016signal}. Nevertheless, as will be seen later, different definition of GFT will not influence the design of our subset selection algorithm.

In the context of GSP,  an original graph signal is usually considered to be generated from a known and fixed class of graph signals with specific properties that are related to the irregular graph structure $\mathcal{G}$. In this paper, we focus on the bandlimited graph signals, as defined in~\cite{10.1093/imaiai/iax021,7472864}. Denote the first $K$ components of $\mathbf{\widehat{x}}$ as $\mathbf{\widehat{x}}_K \in \mathbb{R}^{K}$ and the first $K$ columns of $\mathbf{V}$ as $\mathbf{V}_K \in\mathbb{R}^{N \times K}$. The $K$-bandlimited graph signal is defined as follows. 
\begin{definition}[Bandlimited Graph Signal]
	A graph signal \(\mathbf{x} \in \mathbb{R}^{N}\) is called $K$-bandlimited on a graph
	\(\mathcal{G}=(\mathcal{V}, \mathbf{A})\) with parameter \(K \in\{0,1, \cdots, N-1\},\) when the 	graph frequency components satisfy
	\begin{equation}
	\widehat{x}_{k}=0 \quad \forall k \geq K.
	\end{equation}
\end{definition}

It has been proved in~\cite{chen2016signal} that the bandlimited graph signal has the smoothness property, which builds a bridge between the node domain and spectral domain. 

\setlength{\arrayrulewidth}{1.0pt}
\begin{table}[t!]
	\renewcommand\arraystretch{1.2}
		\begin{center} 	
	\caption{Notations in this paper} 
	\label{table1}
	\vspace{-4mm}
		\begin{tabular}{ccc} 
		\hline  {\bf Description}&\bf{Symbol}  &{\bf Dimension} 		  \\
		\hline
		Set of graph nodes& $\mathcal{V}$ & $N$ \\
		
		Sampling set& $\mathcal{S}$ & /\\
		
		Adjancent matrix& $\mathbf{A}$ &$N\times N$ \\

		Laplacian matrix	& $\mathbf{L}$ &$N\times N$\\

		Eigenvector matrix& $\mathbf{V}$ &$N\times N$\\

		First $K$ columns of $\mathbf{V}$& $\mathbf{V}_K$&$N\times K$\\
				
		Sampling operator according to $\mathcal{S}$& $\boldsymbol{\Psi}$ & $|\mathcal{S}|\times N$ \\ 

		Submatrix $\boldsymbol{\Psi}\mathbf{V}_K$  &$\mathbf{V}_{\mathcal{S}K}$ & $|\mathcal{S}|\times k$ \\

		The $i$-th column of $\mathbf{V}$  &  $\mathbf{v}_{:i}$ & $N\times 1$ \\

		The $i$-th row of $\mathbf{V}$ &  $\mathbf{v}_{i:}$& $1\times N$ \\

		The $(i,j)$-th element of $\mathbf{V}$ & $\mathbf{V}_{ij}$ & /\\

		Graph signal& $\mathbf{x}$& $N\times 1$\\
		
		GFT of graph signal& $\widehat{\mathbf{x}}$& $N\times 1$\\

		Observation signal&	$\mathbf{y}_\mathcal{S}$& $M\times 1$ \\
		
		Recovered GFT &$\mathbf{\widehat{x}}_{K}^{\prime}$ & $K\times 1$ \\

		The largest diagonal element  & $f$ & /\\
		
		Determinant & $|\cdot|$ & /\\
		
		Log base $10$& $\log$ &/ \\ 
		\hline
		\end{tabular}
	\end{center}
\end{table}
\setlength{\arrayrulewidth}{1.0pt}

\subsection{Problem Formulation}
We now describe the sampling theory in GSP. Suppose that we sample $M$ ($M\leq N$) points from the bandlimited graph signal \(\mathbf{x} \in \mathbb{R}^{N}\). Define sampling operator \(\boldsymbol{\Psi}: \mathbb{R}^{N} \mapsto \mathbb{R}^{M}\) as
\begin{equation}
\boldsymbol{\Psi}_{i, j}=\left\{\begin{array}{ll}1, & j=\mathcal{S}_{i} \\ 0, & \text{otherwise }\end{array}\right.,
\end{equation}
where \(\mathcal{S}=\left(\mathcal{S}_{1}, \cdots, \mathcal{S}_{M}\right)\) denotes the sequence of sampled indices, and \(\mathcal{S}_{i} \in\{1,2, \cdots, N\}\) is different from each other. In other words, we consider the non-repetitive sampling, which means that for $i\neq j$, $\mathcal{S}_i \neq \mathcal{S}_j$. Assume that there exists noise when we get the samples. Let \(\mathbf{n} \in \mathbb{R}^{M}\) be the \(i . i . d .\) noise with zero mean and unit variance introduced during sampling, which, in practice, is not necessarily restricted to a particular noise structure. Then, the samples are given by \(\mathbf{x}_{\mathcal{S}}=\boldsymbol{\Psi} \mathbf{x} \in \mathbb{R}^M,\) and the observation model is
	\begin{equation}\label{obser}
		\mathbf{y}_{\mathcal{S}}=\boldsymbol{\Psi} \mathbf{x}+\mathbf{n}=[y_{\mathcal{S}_1},\cdots, y_{\mathcal{S}_M}]^\top.
	\end{equation}
	
There are two major topics in GSP, which can be understood as encoding and decoding  procedures. In the encoding procedure, we need to determine the set of sampling points $\mathcal{S}$ that are most informative. As for the decoding procedure, given the sample points and the corresponding observed sampled signal $\mathbf{y}_\mathcal{S}$, we would like to recover the original signal. In this paper, the second procedure is performed based on the best linear unbiased estimation (BLUE).

\subsection{Subset Selection Methods Based on OED}\label{2.3}
We introduce some existing sampling set selection methods based on OED. For bandlimited graph signals, the observation model~\eqref{obser} can be expressed as 
\begin{equation}
\mathbf{y}_{\mathcal{S}}=\boldsymbol{\Psi} \mathbf{V}_{K} \widehat{\mathbf{x}}_{K}+\mathbf{n}.
\end{equation}
Let \(\mathbf{V}_{\mathcal{S} K}=\boldsymbol{\Psi} \mathbf{V}_{K} \in \mathbb{R}^{M\times K}\). Then, we have
\begin{equation}\label{8}
\mathbf{y}_{\mathcal{S}}=\mathbf{V}_{\mathcal{S} K} \mathbf{\widehat{x}}_K+\mathbf{n}.
\end{equation}
The BLUE of  \(\mathbf{\widehat{x}}_K\) from the observed samples \(\mathbf{y}_{\mathcal{S}}\) is
\begin{equation}\label{blue}
\mathbf{\widehat{x}}_{K}^{\prime}=\mathbf{V}_{\mathcal{S} K}^{\dagger} \mathbf{y}_{\mathcal{S}},
\end{equation}
where \(\mathbf{V}_{\mathcal{S} K}^{\dagger}\) is the pseudo-inverse
of \(\mathbf{V}_{\mathcal{S} K}\). The estimation error of \(\mathbf{x}\) is
	\begin{align}
\mathbf{e}&=\mathbf{V}_K\left(\mathbf{\widehat{x}}_{K}^\prime-\mathbf{\widehat{x}}_{K}\right)\nonumber\\
		&=\mathbf{V}_K\left(\mathbf{V}_{\mathcal{S} K}^{\dagger} y_{\mathcal{S}} - \mathbf{V}_{\mathcal{S} K}^{\dagger} y_{\mathcal{S}} + \mathbf{V}_{\mathcal{S} K}^{\dagger} \mathbf{n}\right)\nonumber\\
		&=\mathbf{V}_{K} \mathbf{V}_{\mathcal{S} K}^{\dagger} \mathbf{n}.
	\end{align}
Thus, the covariance matrix of estimation error is given by
\begin{equation}\label{error}
\mathbf{E}=\mathbb{E}\left[\mathbf{e} \mathbf{e}^{\top}\right]={\mathbf{V}}_{K}\left({\mathbf{V}}_{\mathcal{S} K}^{\top} {\mathbf{V}}_{\mathcal{S} K}\right)^{-1} {\mathbf{V}}_{K}^{\top}.
\end{equation}

In OED, various criteria have been defined to measure the reconstruction error. There have been some works that brought the idea from OED to define an objective function in GSP. For example, if we seek to minimize the mean squared error of covariance, noticing that $\mathbf{V}_K$ is orthogonal, we have
\begin{equation}\label{Tr}
	\begin{split}
		\mathcal{S}&=\underset{|\mathcal{S}|=M}{\arg \min }~\operatorname{Tr(\mathbf{E})} =\underset{|\mathcal{S}|=M}{\arg \min }~\operatorname{Tr}\left({\mathbf{V}}_{\mathcal{S} K}^{\top} {\mathbf{V}}_{\mathcal{S} K}\right)^{-1},\\
	\end{split}
\end{equation}
where $\operatorname{Tr}$ is the trace of a matrix. This is the so-called {A-optimal} design in~\cite{Uncertaintyprincipleandsampling}. Similarly, if we aim to design the optimal sampling set of size $M$ to minimize the worst-case reconstruction error, i.e., the maximum eigenvalue of covariance matrix $\mathbf{E}$, we have
\begin{equation}
\mathcal{S}=\underset{|\mathcal{S}|=M}{\arg \min }~\left\|\mathbf{E}\right\|_2,
=\underset{|\mathcal{S}|=M}{\arg \min }~\left\|\left({\mathbf{V}}_{\mathcal{S} K}^{\top} {\mathbf{V}}_{\mathcal{S} K}\right)^{-1}\right\|_2,
\end{equation}
which is equivalent to {E-optimal} design~\cite{DSPsampling}. To date, however, little therotical performance gurantee for this objective function has been known. In the literature of experimental design, a further design criterion is often considered as a surrogate~\cite{Uncertaintyprincipleandsampling}:
\begin{align}\label{D}
		\mathcal{S}=\underset{|\mathcal{S}|=M}{\arg \min }~\ln \left|\mathbf{E}\right|&=\underset{|\mathcal{S}|=M}{\arg \min }~\ln \left|\left({\mathbf{V}}_{\mathcal{S} K}^{\top} {\mathbf{V}}_{\mathcal{S} K}\right)^{-1}\right|\nonumber\\
		& = \underset{|\mathcal{S}|=M}{\arg \max}~\ln \left|{\mathbf{V}}_{\mathcal{S} K}^{\top} {\mathbf{V}}_{\mathcal{S} K}\right|.
\end{align}
where $|\cdot|$ denotes the determinant of one matrix. The {D-optimal} function is a monotonically increasing submodular function~\cite{shamaiah2010greedy}; thus, it can be used with provable performance guarantees. As shown in~\cite{boyd2004convex}, the {D-optimal} function also gives a quantitative measure of how informative the collection of measurements is, and meanwhile minimizes the volume of the resulting confidence ellipsoid for a fixed confidence level.
 
\section{The Proposed Sampling Method Based on G-optimal}\label{sec44}
In this section, we will introduce our proposed sampling method based on the G-optimal experiment design. To be specific, we will first give some special cases to explain the motivation of our method. Then, we will prove several theoretical results showing the connection between our method and D-optimal and the performance guarantee of our method.
 
\subsection{Sampling with Prior Knowledges}\label{subsec2}
Assume that $n_i$ are independent identically distributed $\mathcal{N}(0,\sigma^2)$ random variables. Under this assumption, the noisy model~\eqref{8} can be rewritten as the following $M$ equations
\begin{align}
	y_{\mathcal{S}_1} &= \mathbf{v}_{\mathcal{S}_1:}\widehat{\mathbf{x}}_{K} + n_1,
	\nonumber\\ &\cdots
	\nonumber\\ y_{\mathcal{S}_M} &= \mathbf{v}_{\mathcal{S}_M:}\widehat{\mathbf{x}}_{K} + n_M,
\end{align}
where $\mathbf{v}_{\mathcal{S}_i:} \in \mathbb{R}^{1 \times K}$ is the ${\mathcal{S}_i}$-th row of $\mathbf{V}_K$.  Then, the maximum-likelihood estimate of $\mathbf{\widehat{x}}_K$ can be given by
\begin{equation}
	\mathbf{\widehat{x}}_K^\prime = \left(\sum_{i=1}^{M}\mathbf{v}_{\mathcal{S}_i:}^\top\mathbf{v}_{\mathcal{S}_i:}\right)^{-1}\sum_{i=1}^{M}y_{\mathcal{S}_i}\mathbf{v}_{\mathcal{S}_i:}^\top.
\end{equation}
Instead of considering the reconstruction error between the recovered signal and the original signal as in~\eqref{error}, we consider the spectral estimation error $\mathbf{e}^\prime=\mathbf{\widehat{x}}_K-\mathbf{\widehat{x}}_K^\prime$, which has zero mean and covariance
\begin{equation}
	\boldsymbol{\Sigma} = \sigma^2 \left(\sum_{i=1}^{M}\mathbf{v}_{\mathcal{S}_M:}^\top\mathbf{v}_{\mathcal{S}_M:}\right)^{-1},
\end{equation}
where the variance of the $i$-th coordinate of the estimation error $\mathbf{e}^\prime$ is $\Sigma_{ii}$. Suppose that the variance $\sigma$ of Gauss distribution is known. Then, the worst case coordinate error variance is the largest diagonal entry of the covariance matrix. In this case, choosing the sampling set to minimize this measure can be expressed as the following problem:
\begin{equation}
\begin{array}{ll}
     \underset{z_1, \cdots, z_n}{\min} & \underset{j=1, \cdots, k}{\max}\left(\left(\sum_{i=1}^{n} z_{i} \mathbf{v}_{i:}^\top \mathbf{v}_{i:}\right)^{-1}\right)_{j j}\\ 
~\\
\text { s.t. } & \sum_{i=1}^n z_i=m, \\ 
~\\
& z_{i} \in\{0,1\}, \quad i=1, \cdots, n.
\end{array}
\end{equation}
Here, $z_i$ can be viewed as an indicator where $z_i = 1$ if the $i$-th element $\mathbf{v}_{i:}^\top \mathbf{v}_{i:}$ is selected and $z_i = 0$ if the $i$-th element is not included. The optimization problem can be solved via greedy search. As shown in~\cite{joshi2008sensor}, this problem can be transformed to:
\begin{equation}\label{20}
     \begin{array}{ll}
     \underset{t, z_1, \cdots, z_n}{\min}  & t \\
     ~\\
     \text{s.t.} &  \left[\begin{array}{cc}t & \mathbf{e}_{j}^{\top} \\ 
     ~ \\
     \mathbf{e}_{j} & \sum_{i=1}^{m} z_{i} \mathbf{v}_{i:}^\top \mathbf{v}_{i:}\end{array}\right] \succeq 0,~j = 1,2,\cdots, k\\
     ~\\
     & \sum_{i=1}^n z_i=m,\\
     ~\\
     & z_{i} \in\{0,1\}, \quad i=1, \cdots, n,
     \end{array}
\end{equation}
where $\mathbf{e}_i \in \mathbb{R}^{k\times 1}$ has $1$ in the $i$-th position and $0$ in other positions. Considering the integer constraint on $z_i$, however, problem~\eqref{20} may be difficult to solve directly. One feasible way is to relax the integer constraint to $0 \leq z_i \leq 1$ and view $z_i$ as the probability of the $i$-th element $\mathbf{v}_{i:}^\top \mathbf{v}_{i:}$ being selected. After solving the program and obtain all solutions $z_i^\prime$, we arrange them and let the $m$ largest ones to be $1$ and the rest to be $0$. In this way, the original problem is solved approximately.

So far, we have worked on sampling with Gauss-noise prior. We next extend this method to Bayesian case where the graph signal $\mathbf{x}$ is drawn from the multivariate normal distribution
\begin{equation}
\quad p\left(\mathbf{x}\right) \propto \exp \left(-\left(\mathbf{x}-\boldsymbol{\mu}\right)^{\top} \boldsymbol{\Sigma}_\mathbf{x}^{-1}\left(\mathbf{x}-\boldsymbol{\mu}\right)\right).
\end{equation}
Without loss of generality,  assume that \(\boldsymbol{\Sigma}_\mathbf{x}\) is full-rank. If \(\boldsymbol{\mu}\) is constant over the vertex set, this prior is analogous to the generative model for a Gaussian wide-sense stationary random process on graphs. So, we do not restrict \(\boldsymbol{\mu}\) to be constant. Remember that $\mathbf{V}_K$ is orthogonal, and thus the vector $\widehat{\mathbf{x}}_K$ also satisfies the multivariate normal distribution with mean $\mathbf{V}_K \boldsymbol{\mu}$ and covariance $\mathbf{V}_K\boldsymbol{\Sigma}_\mathbf{x}\mathbf{V}_K^\top$. Moreover, assume that the noise $n_i$ are independent random Gauss variables, we can derive the error variance of the maximum posteriori probability (MAP) estimate of $\widehat{\mathbf{x}}_K$ as
\begin{equation}
\boldsymbol{\Sigma}_B = \sigma^2\left(\sum_{i=1}^{M}\mathbf{v}_{\mathcal{S}_M:}^\top\mathbf{v}_{\mathcal{S}_M:} + \left(\mathbf{V}_K\boldsymbol{\Sigma}_\mathbf{x}\mathbf{V}_K^\top\right)^{-1}\right)^{-1}.
\end{equation}
A similar approach can be easily obtained to minimize the worst-case coordinate error variance. 

\subsection{The Proposed Sampling Method GOD}
From Section~\ref{subsec2}, we show that through minimizing the maximum diagonal value, we minimize the worst-case error variance of interested prediction error. Hence we give our proposed sampling method based on this idea named \textbf{G}-\textbf{O}ptimal \textbf{D}esign graph sampling (GOD), which has two slightly different forms:
\begin{equation}
	\mathcal{S}=\underset{|\mathcal{S}|=M}{\arg \min }~{\max \operatorname{diag}}~ {\mathbf{V}}_{K}\left({\mathbf{V}}_{\mathcal{S} K}^{\top} {\mathbf{V}}_{\mathcal{S} K}\right)^{-1} {\mathbf{V}}_{K}^{\top},
\end{equation}
or 
\begin{equation}
	\mathcal{S}=\underset{|\mathcal{S}|=M}{\arg \min }~{\max \operatorname{diag} }~\left({\mathbf{V}}_{\mathcal{S} K}^{\top} {\mathbf{V}}_{\mathcal{S} K}\right)^{-1}.
\end{equation}
Here $\max\operatorname{diag}(\mathbf{A})$ means calculating the maximum diagonal element of of the matrix $\mathbf{A}$. For a concise representation, we represent this operation with $f()$.  That is:
\begin{equation}\label{10}
	\mathcal{S}=\underset{|\mathcal{S}|=M}{\arg \min }~f\left( {\mathbf{V}}_{K}\left({\mathbf{V}}_{\mathcal{S} K}^{\top} {\mathbf{V}}_{\mathcal{S} K}\right)^{-1} {\mathbf{V}}_{K}^{\top}\right),
\end{equation}
or 
\begin{equation}\label{11}
	\mathcal{S}=\underset{|\mathcal{S}|=M}{\arg \min }~f\left( \left({\mathbf{V}}_{\mathcal{S} K}^{\top} {\mathbf{V}}_{\mathcal{S} K}\right)^{-1}\right).
\end{equation}
This problem is combinatorial and NP-hard. As with other sampling methods based on optimal experiments design, the greedy method is used to select nodes. In the $k$-step, we choose the next node $j^\star$ satisfying
\begin{equation}
j^\star = \underset{j \in \mathcal{V}\backslash\mathcal{S}_{\text{k-1}}}{\arg\min}~f\left({\mathbf{V}}_{K}\left(\mathbf{V}_{\mathcal{S}_{\text{k-1}}\bigcup\{j\}K}^\top\mathbf{V}_{\mathcal{S}_{\text{k-1}}\bigcup\{j\}K}\right)^{-1}{\mathbf{V}}_{K}^\top\right),
\end{equation}
or
\begin{equation}
j^\star = \underset{j \in \mathcal{V}\backslash\mathcal{S}_{\text{k-1}}}{\arg\min}~f\left((\mathbf{V}_{\mathcal{S}_{\text{k-1}}\bigcup\{j\} K}^\top\mathbf{V}_{\mathcal{S}_{\text{k-1}}\bigcup\{j\}K})^{-1}\right).
\end{equation}
where $\mathcal{S}_{\text{k-1}}$ denotes the sampling set after the ($k-1$)-th step. It is worth saying that the above two forms are not equivalent in general. We choose Eq.~\eqref{11} as our GOD algorithm for its better reconstruction performance and lower computation complexity. The algorithm is summarized in Alg.~\ref{GOD}.
\begin{algorithm}[t!]
	\caption{The GOD Algorithm}
	\label{GOD}
	\begin{algorithmic}
		\Require Graph Laplacian Operator $\mathbf{L}$, sampling budget $M$, sampling set $\mathcal{S} = \emptyset$.		
		\State Compute the eigen-decomposition of $\mathbf{L}$.
		\While {$|\mathcal{S}| \le M$}		
		\State $j^\star = \underset{j \in \mathcal{V}\backslash\mathcal{S}_{\text{k-1}}}{\arg\min}~f\left(\left(\mathbf{V}_{\mathcal{S}_{\text{k-1}}\bigcup\{j\}}^\top\mathbf{V}_{\mathcal{S}_{\text{k-1}}\bigcup\{j\}}\right)^{-1}\right)$		
		\State Update $\mathcal{S} \leftarrow \mathcal{S} \cup \{j^\star\}$.	
		\EndWhile		
		\Ensure $\mathcal{S}$.
	\end{algorithmic}
\end{algorithm}

As discussed in Section~\ref{subsec2}, the above problem can be transformed into semi-definite programming and be solved by any optimization tool like interior-point methods. We try both greedy-based and convex-based methods on a toy model. The graph is a sensor graph with $10$ nodes with $2$-bandlimited graph signals. We compare the reconstruction RMSE under different sampling sizes. The results are averaged over $50$ times. As shown in Fig.~\ref{toy2}, the greedy performs much better than the convex-based method. Also, notice that the convex-based method needs to know the eigen-decomposition of the graph Laplacian $\mathbf{L}$, which is highly complicated to solve. Thus in the following part, we focus on the greedy-based method as Alg.~\ref{GOD}.
\begin{figure}[t!]
	\centering
	\includegraphics[scale = 0.245]{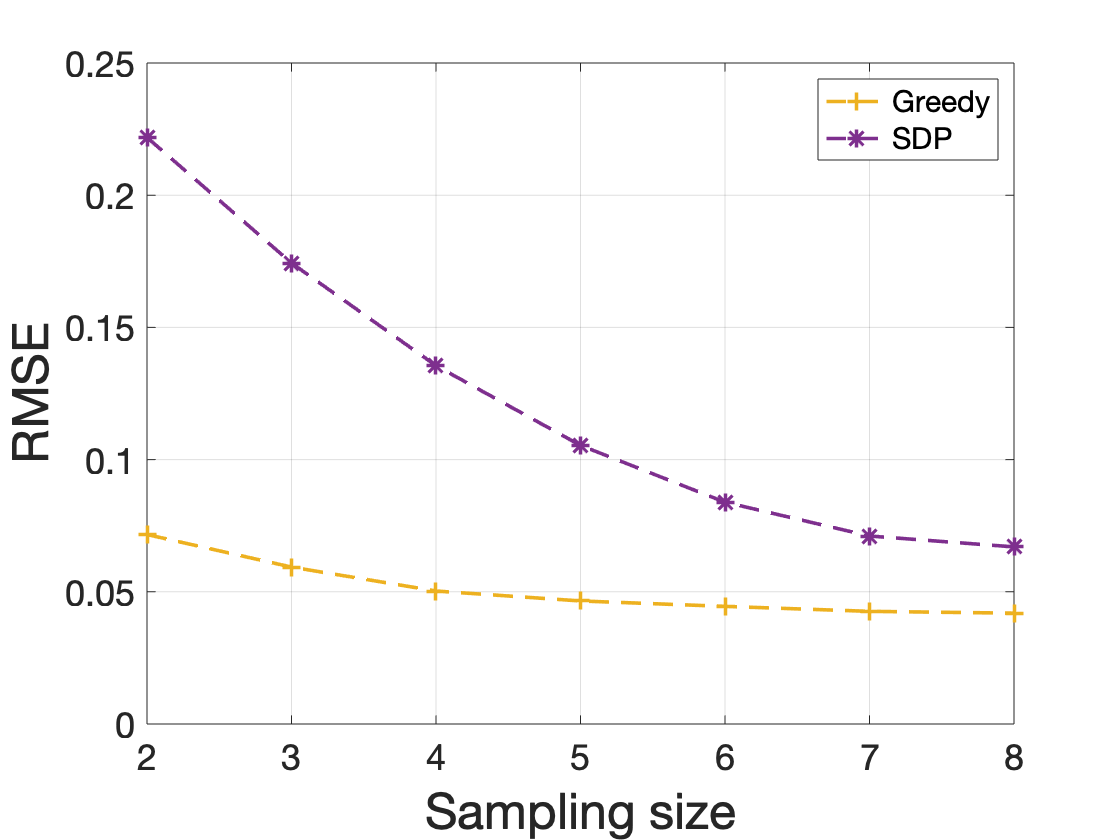}
	\caption{An example of two optimization methods}
	\label{toy2}
\end{figure}

\subsection{Relation Between GOD and D-optimal}
Before giving a theoretical analysis of~\eqref{11}, we first consider an augmented objective by adding a weighted identity matrix , which is named \textbf{A}ug \textbf{G}-\textbf{O}ptimal \textbf{D}esign graph signal sampling (AGOD):
\begin{equation}\label{21}
\mathcal{S}=\underset{|\mathcal{S}|=M}{\arg \min }~f\left(\left({\mathbf{V}}_{\mathcal{S} K}^{\top} {\mathbf{V}}_{\mathcal{S} K} + \mu \mathbf{I}\right)^{-1}\right).
\end{equation}
Notice that if we choose a pretty small $\mu$, problem~\eqref{11} and~\eqref{21} are similar. We show the good performance of our algorithm in this case in the experimental part. The effects of $\mu $ are detailed as follows:
\begin{enumerate}
	\item When the sampling budget is smaller than bandwidth, in other words, $|\mathcal{S}|<K$, the rank of ${\mathbf{V}}_{\mathcal{S} K}^{\top} {\mathbf{V}}_{\mathcal{S} K}$ is $|\mathcal{S}|$. That means the matrix is not full rank. Thus we need to calculate its pseudo-inverse rather than inverse matrix. After adding an identical matrix we can ensure that ${\mathbf{V}}_{\mathcal{S} K}^{\top} {\mathbf{V}}_{\mathcal{S} K} + \mu \mathbf{I}$ is always invertible, which paves a way to solve problems when sampling size is smaller than bandwidth.  
	\item As pointed out in~\cite{8827302}, $\mu $ is a small weight (shift) parameter. If we choose a small $\mu$ (approximate $0.01$ in experiments), Eq.~\eqref{21} is nearly equivalent to Eq.~\eqref{11}. However, small $\mu$ would cause the matrix inverse computation to be inaccurate on a precision-constrained platform. This is because although ${\mathbf{V}}_{\mathcal{S} K}^{\top} {\mathbf{V}}_{\mathcal{S} K} + \mu \mathbf{I}$ is a positive definite matrix, it can potentially be poorly conditioned. Thus in the part of the experiment, we adopt the choice of $\mu $ as suggested in~\cite{8827302}. 
\end{enumerate}

\begin{algorithm}[t!]
	\caption{The AGOD Algorithm }
	\label{AGOD}
	\begin{algorithmic}
		\Require Graph Laplacian Operator $\mathbf{L}$, sampling budget $M$, parameter $\mu$, sampling set $\mathcal{S} = \emptyset$.		
		\State Compute the eigen-decomposition of $\mathbf{L}$.
		\While {$|\mathcal{S}| \le M$}		
		\State $j^\star = \underset{j \in \mathcal{V}\backslash\mathcal{S}_{\text{k-1}}}{\arg\min}~f\left(\left(\mathbf{V}_{\mathcal{S}_{\text{k-1}}\bigcup\{j\}}^\top\mathbf{V}_{\mathcal{S}_{\text{k-1}}\bigcup\{j\}} + \mu \mathbf{I}\right)^{-1}\right)$		
		\State Update $\mathcal{S} \leftarrow \mathcal{S} \cup \{j^\star\}$.	
		\EndWhile		
		\Ensure $\mathcal{S}$.
	\end{algorithmic}
\end{algorithm}

We adopt a greedy algorithm named AGOD summarized in Alg.~\ref{AGOD} as discussed before. We then give an intuitive performance discussion about the performance of AGOD, which established a connection with the D-optimal sampling method as follows.
\begin{proposition}\label{proposition1}
    Minimizing the objective function $f\left(\left({\mathbf{V}}_{\mathcal{S} K}^{\top} {\mathbf{V}}_{\mathcal{S} K} + \mu \mathbf{I}\right)^{-1}\right)$ is equivalent to minimizing the upper bound of an augment objective function of D-optimal $\frac{1}{K}\ln \left|\left({\mathbf{V}}_{\mathcal{S} K}^{\top} {\mathbf{V}}_{\mathcal{S} K} + \mu \mathbf{I}\right)^{-1}\right|$.
\end{proposition}
\begin{proof}
    The detailed proof is shown in Appendix~\ref{prolemma1}.
\end{proof}

We point out that minimizing the objective function $f\left( {\mathbf{V}}_{K}\left({\mathbf{V}}_{\mathcal{S} K}^{\top} {\mathbf{V}}_{\mathcal{S} K}\right)^{-1} {\mathbf{V}}_{K}^{\top}\right)$ is also equivalent to minimizing the upper bound of an augment objective function of D-optimal $\frac{1}{K}\ln \left|\left({\mathbf{V}}_{\mathcal{S} K}^{\top} {\mathbf{V}}_{\mathcal{S} K} + \mu \mathbf{I}\right)^{-1}\right|$. It is obvious considering
	\begin{align}
		 &\left|\left({\mathbf{V}}_{K}\left({\mathbf{V}}_{\mathcal{S} K}^{\top} {\mathbf{V}}_{\mathcal{S} K}+\mu\mathbf{I}\right)^{-1} {\mathbf{V}}_{K}^{\top}\right)\right|\nonumber\\
		 =&\left|{\mathbf{V}}_{K}\right|\left|\left({\mathbf{V}}_{\mathcal{S} K}^{\top} {\mathbf{V}}_{\mathcal{S} K}+\mu\mathbf{I}\right)^{-1}\right| \left|{\mathbf{V}}_{K}^{\top}\right|\nonumber\\
		 =&\left|\left({\mathbf{V}}_{\mathcal{S} K}^{\top} {\mathbf{V}}_{\mathcal{S} K}+\mu\mathbf{I}\right)^{-1}\right|.
	\end{align}
The last equation holds because ${\mathbf{V}}_{K}{\mathbf{V}}_{K}^{\top} = \mathbf{I}$. Proposition~\ref{proposition1} illustrates the relationship between our proposed AGOD and a variant of D-optimal. To show the gap of two objective values in Proposition~\ref{proposition1}, we simulate some toy experiments on a sensor graph with $120$ nodes in Fig.~\ref{MLE}. We conduct with sampling budgets between $5$ and $40$. As shown, the objective function values of both methods (named D-D and G-G) go down as the sampling size increases. We also calculate $\frac{1}{K}\ln \left|\left({\mathbf{V}}_{\mathcal{S} K}^{\top} {\mathbf{V}}_{\mathcal{S} K} + \mu \mathbf{I}\right)^{-1}\right|$ with nodes sampled through our proposed AGOD algorithm (named G-D). It can be found that our proposed can also decrease the objective function value of D-optimal and the gap between G-D and D-D is small, especially with a slightly large sampling budget.

\begin{figure}[t!]
	\centering
	\includegraphics[scale = 0.245]{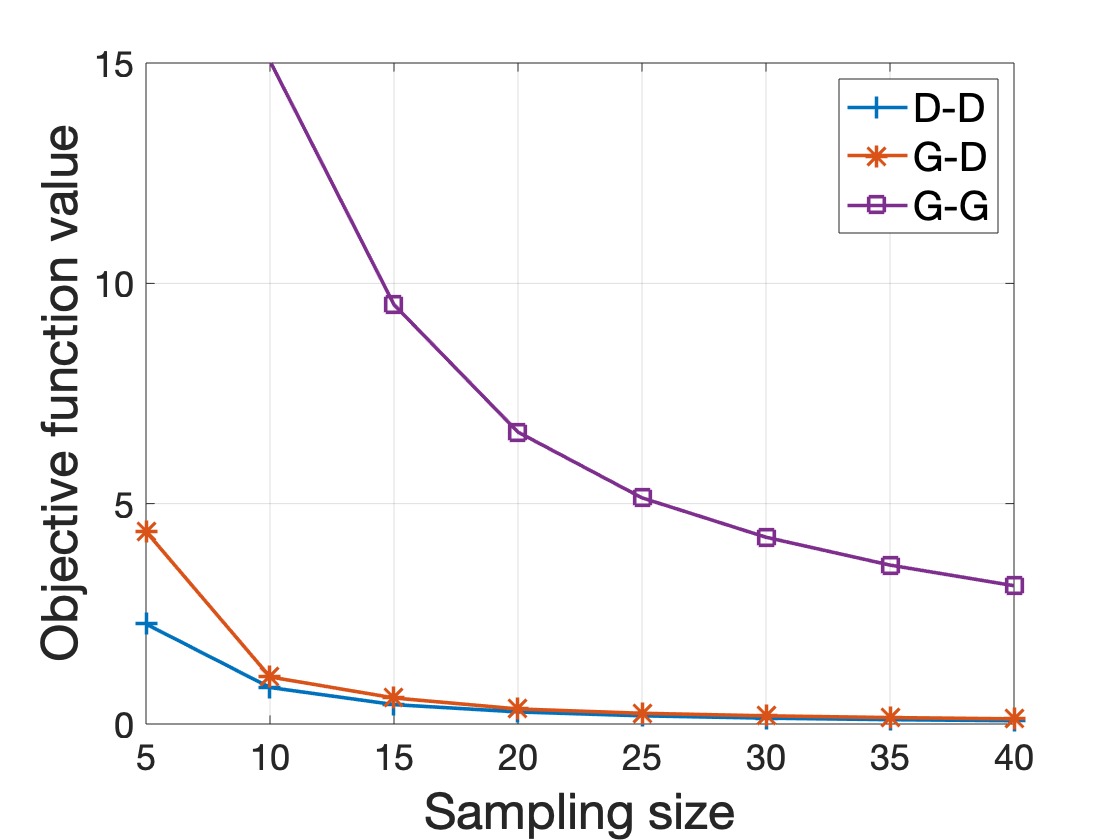}
	\caption{Objective function values gaps under different sampling size. G-G means calulating $f\left(\left({\mathbf{V}}_{\mathcal{S} K}^{\top} {\mathbf{V}}_{\mathcal{S} K} + \mu \mathbf{I}\right)^{-1}\right)$ with nodes sampled through AGOD algorithm. G-D means calulating $\frac{1}{K}\ln \left|\left({\mathbf{V}}_{\mathcal{S} K}^{\top} {\mathbf{V}}_{\mathcal{S} K} + \mu \mathbf{I}\right)^{-1}\right|$ with nodes sampled through AGOD algorithm. D-D means calulating $\frac{1}{K}\ln \left|\left({\mathbf{V}}_{\mathcal{S} K}^{\top} {\mathbf{V}}_{\mathcal{S} K} + \mu \mathbf{I}\right)^{-1}\right|$ with nodes sampled through D-optimal~\cite{Uncertaintyprincipleandsampling}. }
	\label{MLE}
\end{figure}

\subsection{Performance Analyses of \emph{AGOD}}
Although Lemma 1 gives an intuitional understanding of AGOD, it does not provide a performance guarantee. To do that, this section derives near-optimality results that hold for the proposed algorithm based on supermodularity theory.

Supermodularity is defined for functions defined over subsets of a larger set. Intuitively, supermodularity captures the “diminishing returns” property of given functions, which leads to bounds on the suboptimality of their greedy minimization~\cite{nemhauser1978analysis}. Common supermodular functions include the rank, $\log\det$, or Von-Neumann entropy of a matrix~\cite{bach2013learning}. The relative suboptimality value is then defined as follows to evaluate different solutions. 
\begin{definition}
	Let \(g^{*}=g\left(\mathcal{S}^{*}\right)\) be the optimal value of problem
	\begin{equation}\label{44}
	\min _{\mathcal{S} \subseteq \mathcal{V},|\mathcal{S}|=M} g(\mathcal{S}),
	\end{equation}
	and \(\hat{\mathcal{S}}\) is the solution of one algorithm with \(|\hat{\mathcal{S}}|=M\), then the
	relative suboptimality is defined as
	\begin{equation}\label{32}
		r=\frac{g(\hat{\mathcal{S}})-g^{*}}{g(\varnothing)-g^{*}}.
	\end{equation}
\end{definition}
If the objective function \(g(\mathcal{S})\) is monotonic decreasing and supermodular, the value of \(r\) of greedy solution is upper-bounded by \(e^{-1}\). However, supermodularity is a strict condition. To study the proposed objective function, we introduce a general definition of approximate supermodularity~\cite{8047995}.
\begin{definition}
	A set function \(g: 2^{\mathcal{V}} \rightarrow \mathbb{R}\) is \(\alpha\)-supermodular if for all sets
\(\mathcal{A} \subseteq \mathcal{B} \subseteq \mathcal{V}\) and all \(u \notin \mathcal{B}\) it holds that
\begin{equation}\label{supermodular}
g(\mathcal{A} \cup\{u\})-g(\mathcal{A}) \leq \alpha[g(\mathcal{B} \cup\{u\})-g(\mathcal{B})]
\end{equation}
for \(\alpha \geq 0\).
\end{definition}
For \(\alpha \geq 1\), Eq.~\eqref{supermodular} is equivalent to the traditional definition of supermodularity, in which case we refer to the function simply as supermodular/submodular~\cite{bach2013learning}. For \(\alpha \in[0,1)\), however, \(g\) is said to be approximately supermodular/submodular. Notice that~\eqref{supermodular} always holds for \(\alpha=0\) if \(g\) is monotone decreasing. Indeed, \(g(\mathcal{A} \cup\{u\})-g(\mathcal{A}) \leq 0\) in this case. Thus, \(\alpha\)-supermodularity is only of interest when \(\alpha\) takes the largest value for which~\eqref{supermodular} holds, i.e.,
\begin{equation}
\alpha=\min _{\mathcal{A} \subseteq \mathcal{B} \subseteq \mathcal{V}, j \notin \mathcal{B}} \frac{g(\mathcal{A} \cup\{j\})-g(\mathcal{A})}{g(\mathcal{B} \cup\{j\})-g(\mathcal{B})}.
\end{equation}

Given $\alpha$, the following lemma provide the theoretical guarantee of the solution of problem~\eqref{44}~\cite{8047995}:
\begin{lemma}\label{lemma2}
	Let \(g^{*}=g\left(\mathcal{S}^{*}\right)\) be the optimal value of problem~\eqref{44} and \(\mathcal{S}_{l}\) is the set obtained via greedy algorithm with \(l\) samples. If function \(g\) is (i) monotone decreasing and (ii) \(\alpha\)-supermodular, then
	\begin{equation}
	\frac{g\left(\mathcal{S}_{l}\right)-g^{*}}{g(\varnothing)-g^{*}} \leq\left(1-\frac{\alpha}{M}\right)^{l} \leq e^{-\alpha l / M}.
	\end{equation}
\end{lemma}

We now give our main theorem:
\begin{theorem}\label{mainTh}
	The set function \(g(\mathcal{S})\) defined as
	\begin{equation}
    	\begin{split}
    		g(\mathcal{S}) &= f \left({\mathbf{V}}_{\mathcal{S} K}^{\top} {\mathbf{V}}_{\mathcal{S} K} + \mu \mathbf{I}\right)^{-1}\\
    		 &= {\max \operatorname{diag} }\left({\mathbf{V}}_{\mathcal{S} K}^{\top} {\mathbf{V}}_{\mathcal{S} K} + \mu \mathbf{I}\right)^{-1},\\
    	\end{split}
    \end{equation}
	is (i) monotone decreasing and (ii) \(\alpha\)-supermodular with
	\begin{equation}
	\alpha \geq \frac{\mu(2+\mu)}{(1+\mu)^2}.
	\end{equation}
\end{theorem}

\begin{proof}
	The detailed proof is shown in Appendix~\ref{protheorem}.
\end{proof}
In~\cite{8827302}, authors proposed a set function $$g(\mathcal{S})=\operatorname{Tr}\left({\mathbf{V}}_{\mathcal{S} K}^{\top} {\mathbf{V}}_{\mathcal{S} K} + \mu \mathbf{I}\right)^{-1},$$ which is also monotone decreasing and $\alpha-$supermodular with parameter $\alpha$ satisfying 
\begin{equation}
	\alpha \geq \frac{\mu^{3}(2+\mu)}{(\mu+1)^{4}}.
\end{equation}
It is obvious that our proposed set function has a larger lower bound of $\alpha$. According to Lemma~\ref{lemma2}, with $\mathcal{S}$ increases, our sample set can get closer to the optimal solution faster. Parameter \(\mu\) mimics the 'SNR' in~\cite{8047995} via 'SNR' \(=-10 \log \mu\). We illustrate the lower bound of \(\alpha\) in \((23)\) in terms of the value of \(-10 \log \mu\) in Fig.~\ref{FIG1}. This figure matches the result demonstrated in~\cite{8047995}. We can also see that our bound is larger than~\cite{8827302} (named Tr) with a relative small $\mu$, which indicates that AGOD performs better in theory. 
\begin{figure}[t!]
	\includegraphics[scale = 0.245]{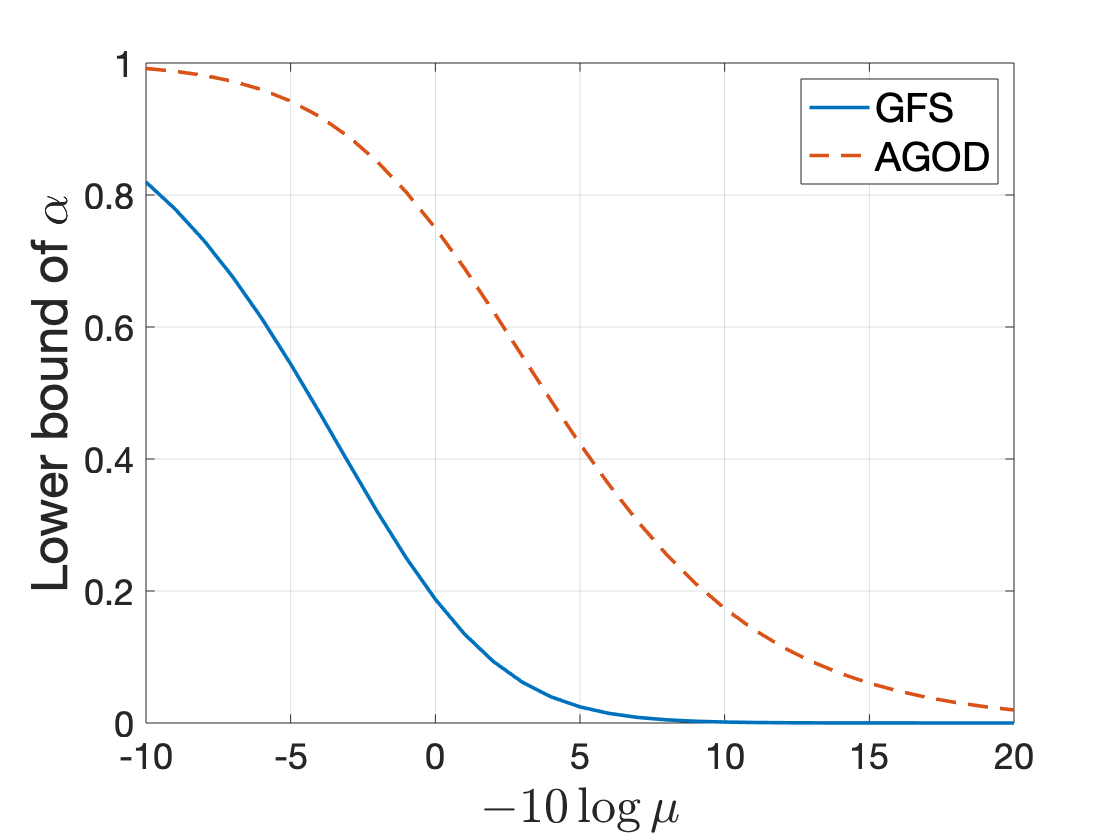}
	\centering
	\caption{Lower bound of $\alpha$ in terms of $ -10\log\mu$}
	\label{FIG1}
\end{figure}

\section{Fast Sampling Method Based on G-optimal}\label{sec5}
As discussed earlier, graphs considered in most real applications are very large. Hence, computing and storing the graph Fourier basis explicitly is often practically infeasible. To alleviate this computation burden, in this section, we propose a fast algorithm without eigen-decomposition operation.

We define the objective function of our fast sampling method, which we call it \textbf{F}ast \textbf{A}ugment \textbf{G}-\textbf{O}ptimal \textbf{D}esign (FAGOD):
	\begin{equation}\label{27}
		\mathcal{S}=\underset{|\mathcal{S}|=M}{\arg \min }~f \left(\left( {\mathbf{V}}_{\mathcal{S} K}{\mathbf{V}}_{\mathcal{S} K}^{\top} + \mu \mathbf{I}\right)^{-1}\right),
	\end{equation}
 where $f()$ means calculating the maximum diagonal element. 

\begin{figure}[h!]
	\includegraphics[scale = 0.245]{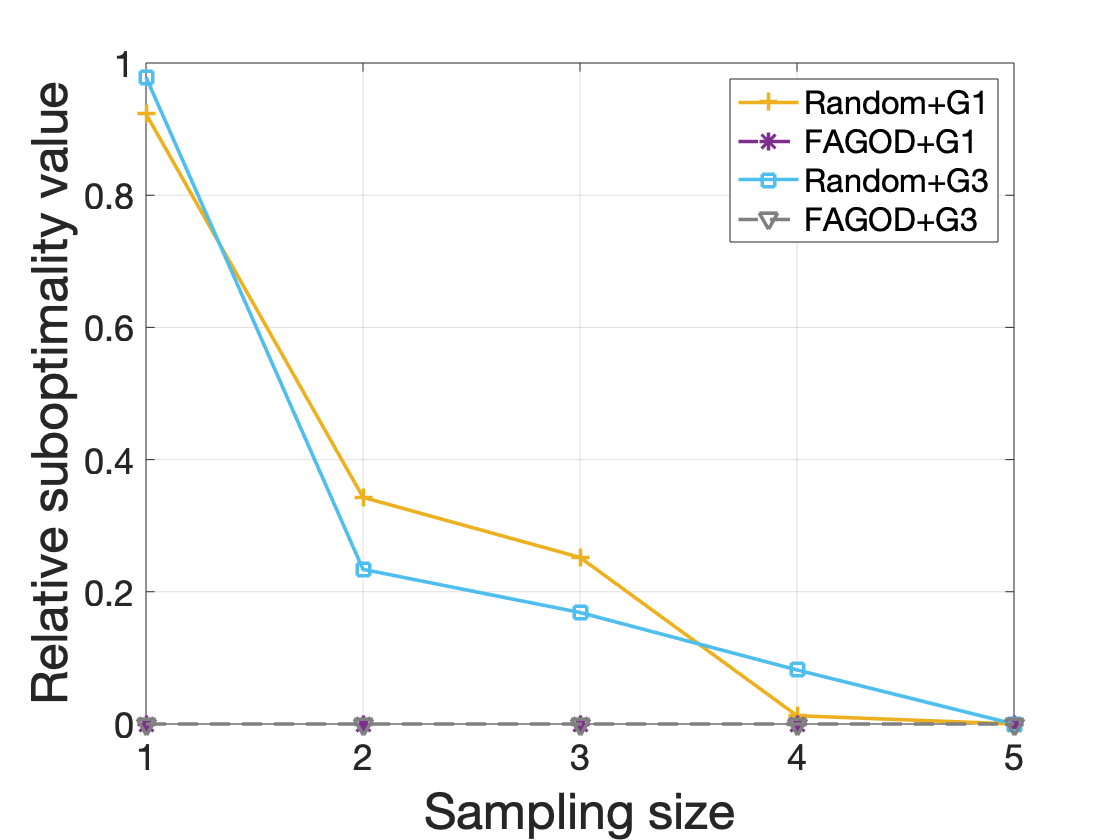}
	\centering
	\caption{Relative suboptimality value in terms of sample size. Here G1 indicates random sensor graph and G3 indicates community graph, which both follow the same settings introduced in Section~\ref{sec6}.}
	\label{Fig2}
\end{figure}

Here ${\mathbf{V}}_{\mathcal{S} K}{\mathbf{V}}_{\mathcal{S} K}^{\top}$ can be viewed as sampling from an idea low-pass graph filter ${\mathbf{V}}_{ K}{\mathbf{V}}_{ K}^{\top}$ with cutoff frequency $K$. However, optimizing the new problem~\eqref{27} still requires computation of an ideal LP filter rather than first $K$ eigenvectors of $\mathbf{L}$. We adopt the method proposed in~\cite{le2017approximate}, which apply Jacobi eigenvalue algorithm to approximate ${\mathbf{V}}_{ K}{\mathbf{V}}_{ K}^{\top}$. This efficient approximation required lower complexity than eigen-decomposition. Specifically, the goal is to solve the following optimization problem
	\begin{equation}\label{40}
		\min _{\widehat{\boldsymbol{\Lambda}}, \mathbf{S}_{1}, \ldots, \mathbf{S}_{J}}\left\|\mathbf{L}-\mathbf{S}_{1} \ldots \mathbf{S}_{J} \widehat{\boldsymbol{\Lambda}} \mathbf{S}_{J}^{\top} \ldots \mathbf{S}_{1}^{\top}\right\|_{F}^{2},
	\end{equation}
where \(\widehat{\boldsymbol{\Lambda}}\) is constrained to be a diagonal matrix, and \(\mathbf{S}_{i}\) are constrained to be {\it Givens} rotation matrices. After solving~\eqref{40}, we obtain an estimate of graph filter ${\mathbf{V}}_{ K}{\mathbf{V}}_{ K}^{\top}$ as $\widetilde{{\mathbf{V}}}_{ K}\widetilde{{\mathbf{V}}}_{ K}^{\top}$, where $\widetilde{{\mathbf{V}}}_{ K}=\mathbf{S}_{1} \ldots \mathbf{S}_{J}$. The advantages of solving~\eqref{40} instead of using a Chebyshev matrix polynomial of $\mathbf{L}$ can be found in~\cite{8827302}. See Algorithm~\ref{Al2} for the whole process.

\begin{algorithm}[t!]
	\caption{The FAGOD Algorithm}
	\label{Al2}
	\begin{algorithmic}
		\Require Graph Laplacian Operator $\mathbf{L}$, sampling budget $M$, parameter $\mu$, sampling set $\mathcal{S} = \emptyset$.
		
		\State Compute low-pass graph filter $\mathbf{T}^{\text {Appro }} = \widetilde{{\mathbf{V}}}_{ K}\widetilde{{\mathbf{V}}}_{ K}^{\top}$.
		\While {$|\mathcal{S}| \le M$}
		
		\State $j^\star = \underset{j \in \mathcal{V}\backslash\mathcal{S}}{\arg\min}~f(\mathbf{T}^{\text {Appro }}_{\mathcal{S}\bigcup{j}}+\mu\mathbf{I})^{-1}$.
		
		\State Update $\mathcal{S} \leftarrow \mathcal{S} \cup \{j^\star\}$.
		
		\EndWhile
		
		\Ensure $\mathcal{S}$.
	\end{algorithmic}
\end{algorithm}

\begin{figure*}[t!]
	\centering 
\subfigure[]{
		\begin{minipage}[t]{0.3\linewidth}			\includegraphics[width=1.2\linewidth]{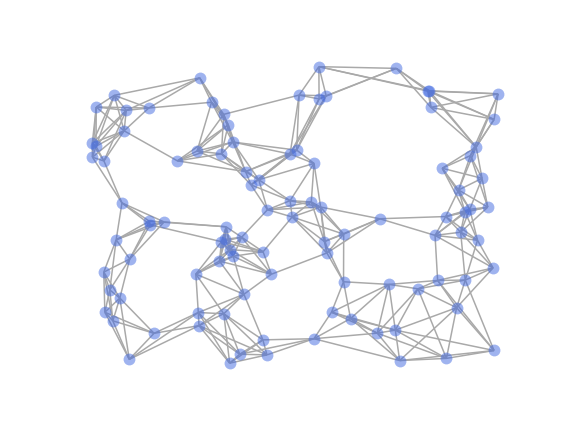}
		\end{minipage}
	} 
	\subfigure[]{
		\begin{minipage}[t]{0.3\linewidth}
				\includegraphics[width=1.2\linewidth]{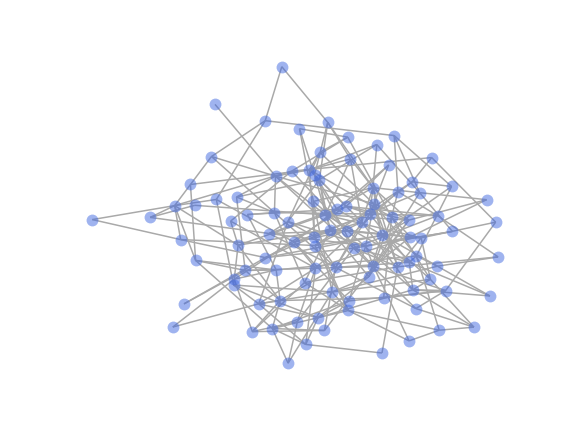}
		\end{minipage}%
	}
	\subfigure[]{
		\begin{minipage}[t]{0.3\linewidth}
			\includegraphics[width=1.2\linewidth]{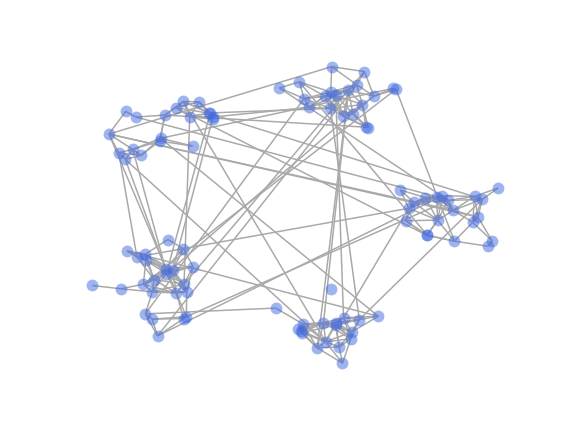}
		\end{minipage}
	}%
	\caption{Illustrations of three types of graph (a) Random sensor graph \textbf{G1}, (b) ER graph \textbf{G2}, (c) Community graph \textbf{G3}.}
		\label{graph}
\end{figure*}

Similarly, to give an intuitional understanding of FAGOD, the following proposition demonstrates the connection between FAGOD and D-optimal:
\begin{proposition}\label{proposition2}
Minimizing the objective function $f \left(\left( {\mathbf{V}}_{\mathcal{S} K}{\mathbf{V}}_{\mathcal{S} K}^{\top} + \mu \mathbf{I}\right)^{-1}\right)$ is equivalent to minimizing the upper bound of an augment objective function of D-optimal $\frac{1}{K}\ln \left|\left({\mathbf{V}}_{\mathcal{S} K}^{\top} {\mathbf{V}}_{\mathcal{S} K} + \mu \mathbf{I}\right)^{-1}\right|$.

\end{proposition}

\begin{proof}
	The detailed proof is shown in Appendix~\ref{pro2}.
\end{proof}

Furthermore, to show the near-optimality of the proposed FAGOD algorithm, we conduct some toy experiments on small graphs ($N = 10$) to compute the exact relative suboptimality value $r$ according to Eq.~\eqref{32} of the solution from Algorithm~\ref{Al2}. We find the optimal value through an exhaustive method. For comparison, we also calculate the $r$ value of random sampling, which selects nodes randomly for each sample size. Parameter setting is the same as Section~\ref{sec6}. Recall that \begin{equation}
		r=\frac{g(\hat{\mathcal{S}})-g^{*}}{g(\varnothing)-g^{*}}.
	\end{equation}

	From Fig.~\ref{Fig2} we can see that the relative suboptimality values of our proposed FAGOD method are $0$ at each sample size, which means that the performance of FAGOD is close to the exhaustive optimal solution. Also the resulted value $r$ is much smaller than that of the random solution.

To avoid eigendecomposition in the reconstruction stage, we adopt a biased signal reconstruction method proposed in~\cite{8827302} and more reconstruction error analysis can be found in~\cite{8827302}.

\begin{lemma}
Adding a parameter $\mu$ in~\eqref{21}, we obtain a biased estimate 
\begin{equation}
\mathbf{\widehat{x}} = \mathbf{V}_K\left({\mathbf{V}}_{\mathcal{S} K}^{\top} {\mathbf{V}}_{\mathcal{S}K} + \mu \mathbf{I}\right)^{-1}\mathbf{V}_{\mathcal{S}K}^\top\mathbf{y}_\mathcal{S},
\end{equation}
which can be approximated as
\begin{equation}
\mathbf{\widehat{x}} = \mathbf{V}_{\mathcal{S}K}\mathbf{V}_{\mathcal{S}K}^{\top}\left({\mathbf{V}}_{\mathcal{S}K} {\mathbf{V}}_{\mathcal{S} K}^{\top} + \mu \mathbf{I}\right)^{-1}\mathbf{y}_\mathcal{S}.
\end{equation}
\end{lemma}

\section{Experiments}\label{sec6}
\subsection{Experimental Setup}
We demonstrate the efficacy of our proposed FAGOD sampling algorithm via extensive simulations. All algorithms are performed on MATLAB R2019b. 

We use three types of graphs in GSPBOX~\cite{perraudin2014gspbox} for testing as shown in Fig.~\ref{graph}:
\begin{enumerate}
	\item[1)] \textbf{G1}: Random sensor graph with $N$ nodes. The graph is a weighted sparse graph. Each node is connected to its six nearest neighbours with weights
	\item[2)] \textbf{G2}: Random unweighted graph with Erd{\"o}s-R\'{e}nyi model (ER graph): The
edge connecting probability was set to \(0.05\) with $N$ nodes.
	\item[3)] \textbf{G3}: Community graph with $N$ nodes and random $\lfloor\sqrt{N} / 2\rfloor$ communities.
\end{enumerate}

We set $N = 400$ for all three graphs. The performance of the sampling methods depends on the assumptions about the true signal and sampling noise. We consider the problem in the following scenarios:

\begin{enumerate}
	\item[1)]\textbf{GS1}: The true signals are exactly \(K\)-bandlimited, where $K = 10$. The non-zero GFT coefficients are randomly generated from \(\mathcal{N}(0, 0.5) .\) 
	\item[2)]\textbf{GS2}: The true signals are approximately bandlimited, where the first $K=10$ GFT coefficients are randomly generated from \(\mathcal{N}(0, 0.5)\), and the left $N-K$ GFT coefficients are randomly generated from \(\mathcal{N}\left(0,5\times 10^{-3}\right)\).
	\item[3)]\textbf{GS3}: The true signals are exactly \(K\)-bandlimited, where $K = 40$. The non-zero GFT coefficients are randomly generated from \(\mathcal{N}(0, 0.5) .\) 
\end{enumerate}

We then add i.i.d. Gaussian noise generated from \(\mathcal{N}\left(0,5\times 10^{-3}\right)\). The performance of the proposed method is compared with the following approaches:
\begin{enumerate}
	\item[1)] Random graph sampling with non-uniform probability distribution~\cite{Randomsampling}.
	\item[2)] Deterministic sampling set selection methods, including MinSpec~\cite{DSPsampling}, MaxCutOff~\cite{Efficientsamplingsetselection}, MinFrob~\cite{Uncertaintyprincipleandsampling}, MaxPVol~\cite{Uncertaintyprincipleandsampling}, GFS~\cite{8827302} and Eigen-decomposition-free (Ed\_Free)~\cite{Eigendecomposition-free}, respectively.
\end{enumerate}

We generate $150$ signals from each of the three signal models on each of the graphs, use the sampling sets obtained from all the methods to perform reconstruction, and plot the mean of the mean squared error (MSE) for different sizes of sampling sets. For FAGOD sampling, the shift parameter $\mu$ is set to be $1/(\kappa_0-1)$, where we set $\kappa_0 = 100$ as the condition number constraint. The number of {it Givens} rotations matrices is $J = 6N \log N$. All the competing algorithms are run with the default or the recommended settings in the referred publications. Finally, we estimate the root mean square error (RMSE) between the reconstructed signal $\mathbf{x}^*$ and the original noiseless signal $\mathbf{x}$ as:
\begin{equation}
	\text{RMSE} = \sqrt{\frac{\|\mathbf{x}^*-\mathbf{x}\|_2^2}{N}}.
\end{equation}

\subsection{RMSE vs. Sampling Size}
\begin{figure*}[t!]
	\centering
	\subfigure[]{
		\begin{minipage}[t]{0.3\linewidth}			\includegraphics[width=1\linewidth]{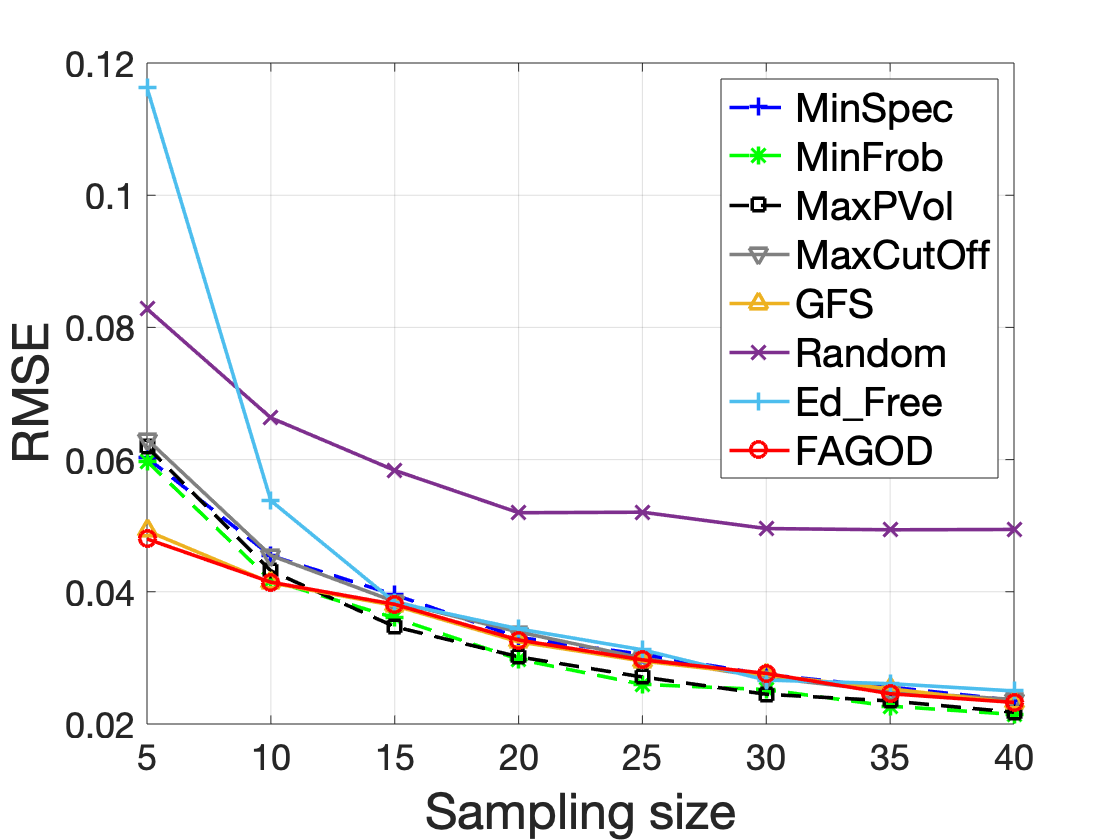}
		\end{minipage}
	}
	\subfigure[]{
		\begin{minipage}[t]{0.3\linewidth}
				\includegraphics[width=1\linewidth]{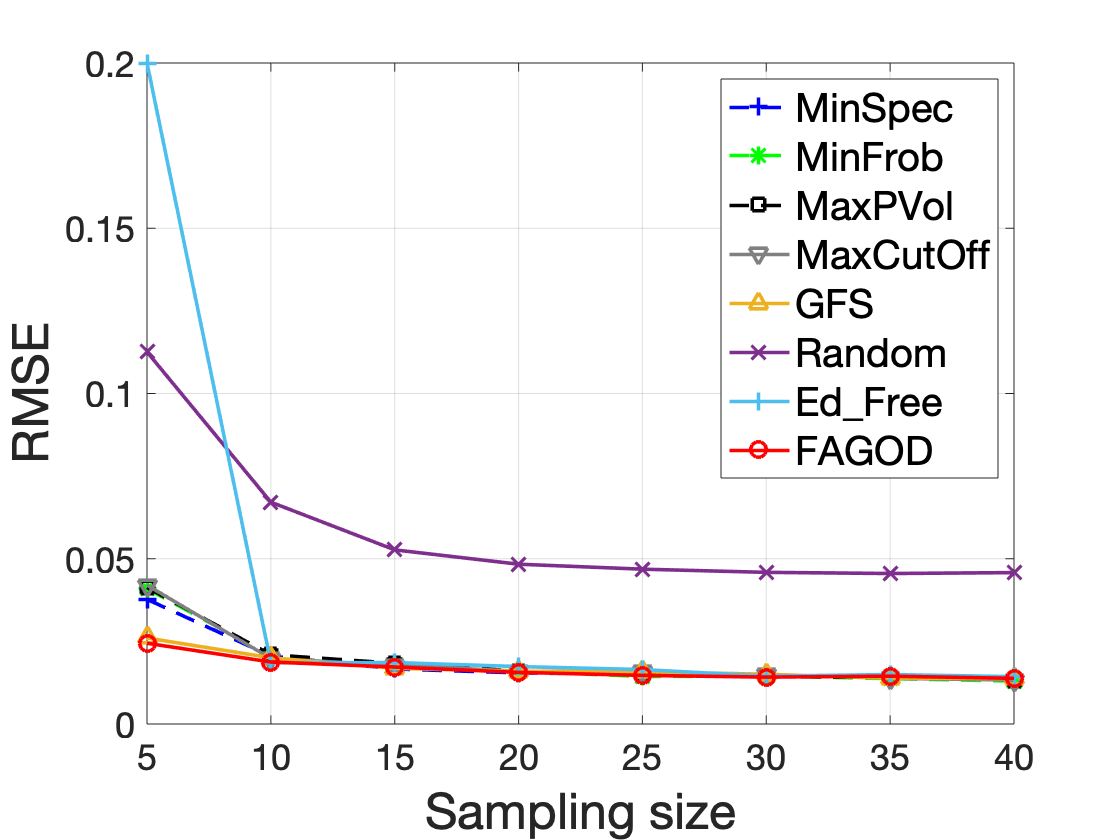}
		\end{minipage}%
	}
	\subfigure[]{
		\begin{minipage}[t]{0.3\linewidth}
			\includegraphics[width=1\linewidth]{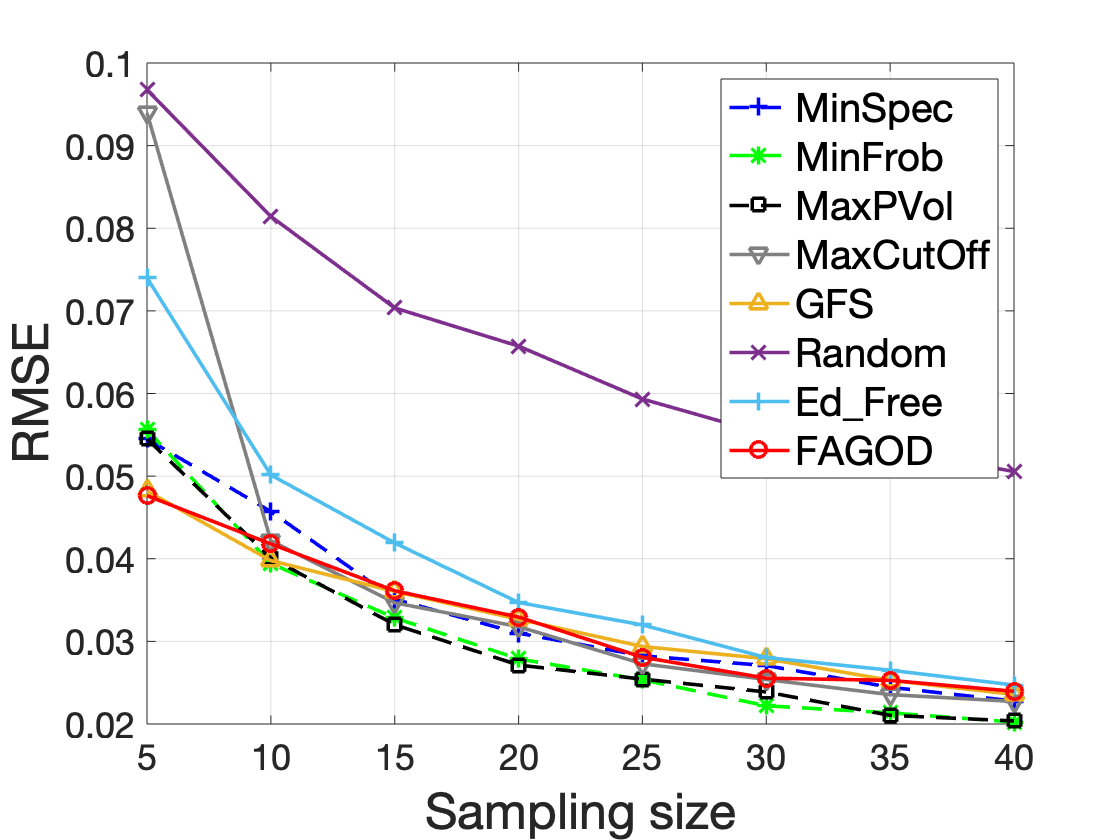}
		\end{minipage}
	}%
	
	\subfigure[]{
		\begin{minipage}[t]{0.3\linewidth}			\includegraphics[width=1\linewidth]{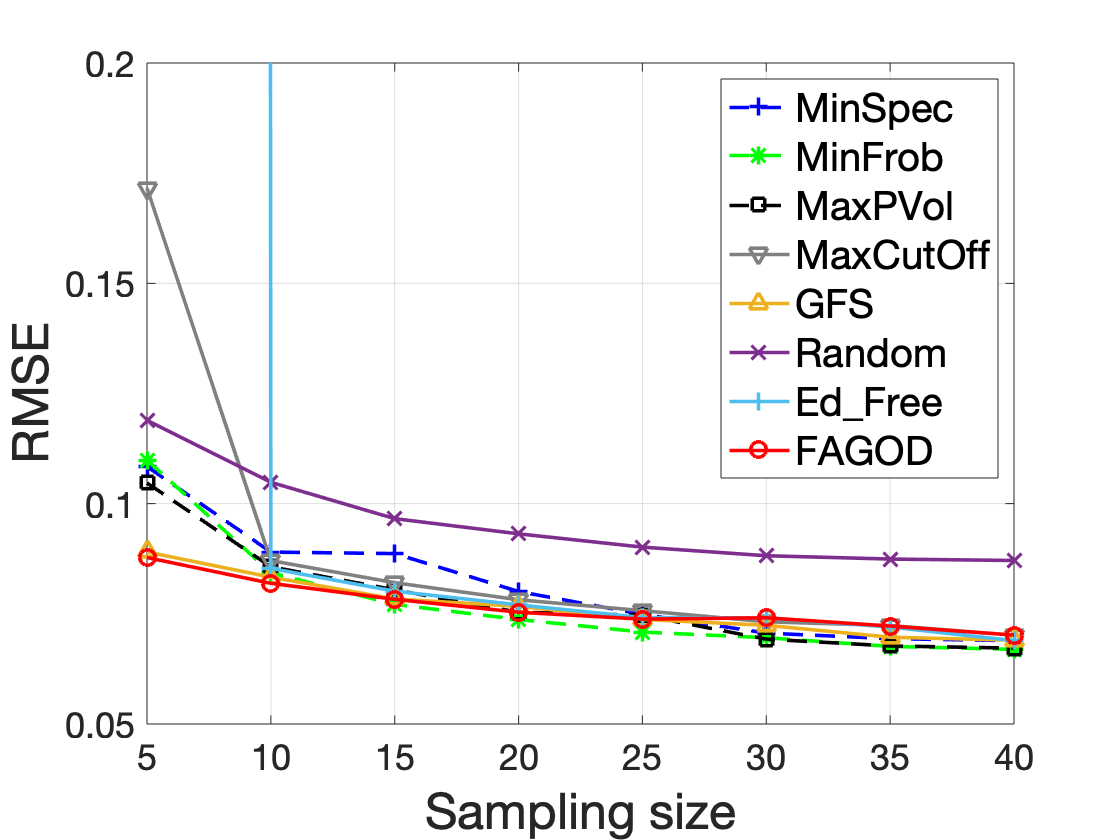}
		\end{minipage}
	}
	\subfigure[]{
		\begin{minipage}[t]{0.3\linewidth}			\includegraphics[width=1\linewidth]{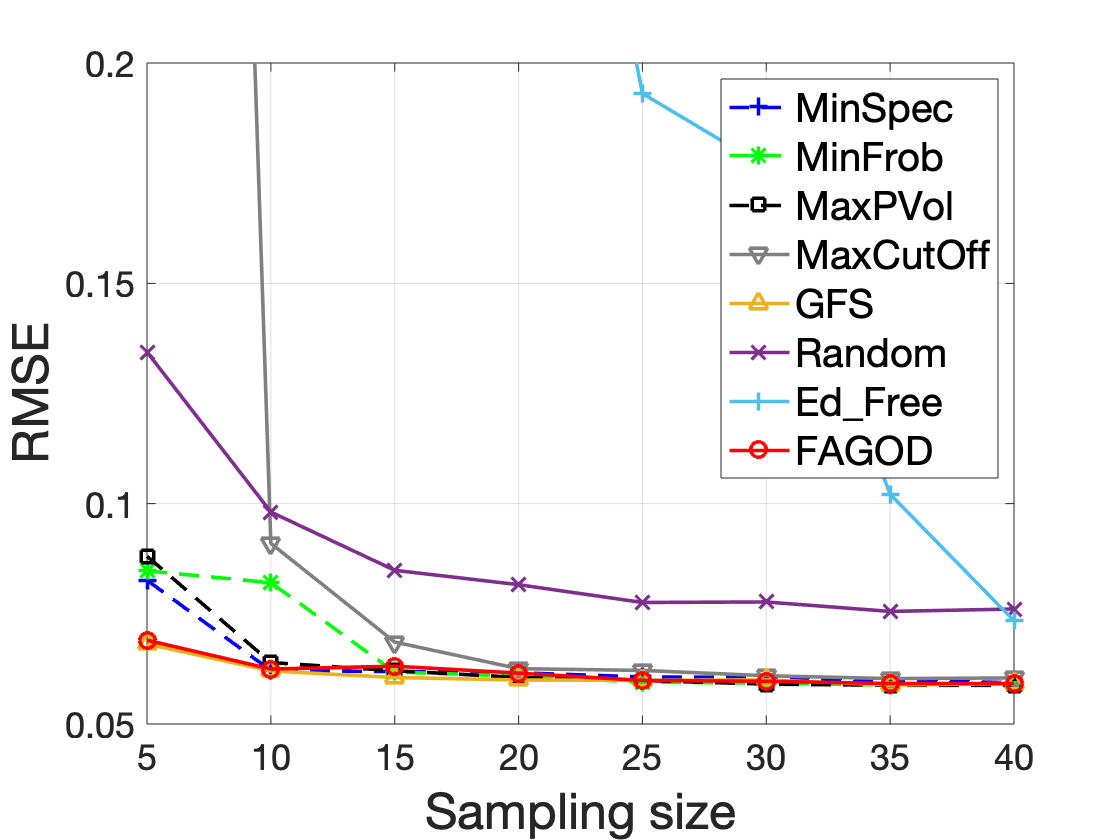}
		\end{minipage}
	}
	\subfigure[]{
		\begin{minipage}[t]{0.3\linewidth}			\includegraphics[width=1\linewidth]{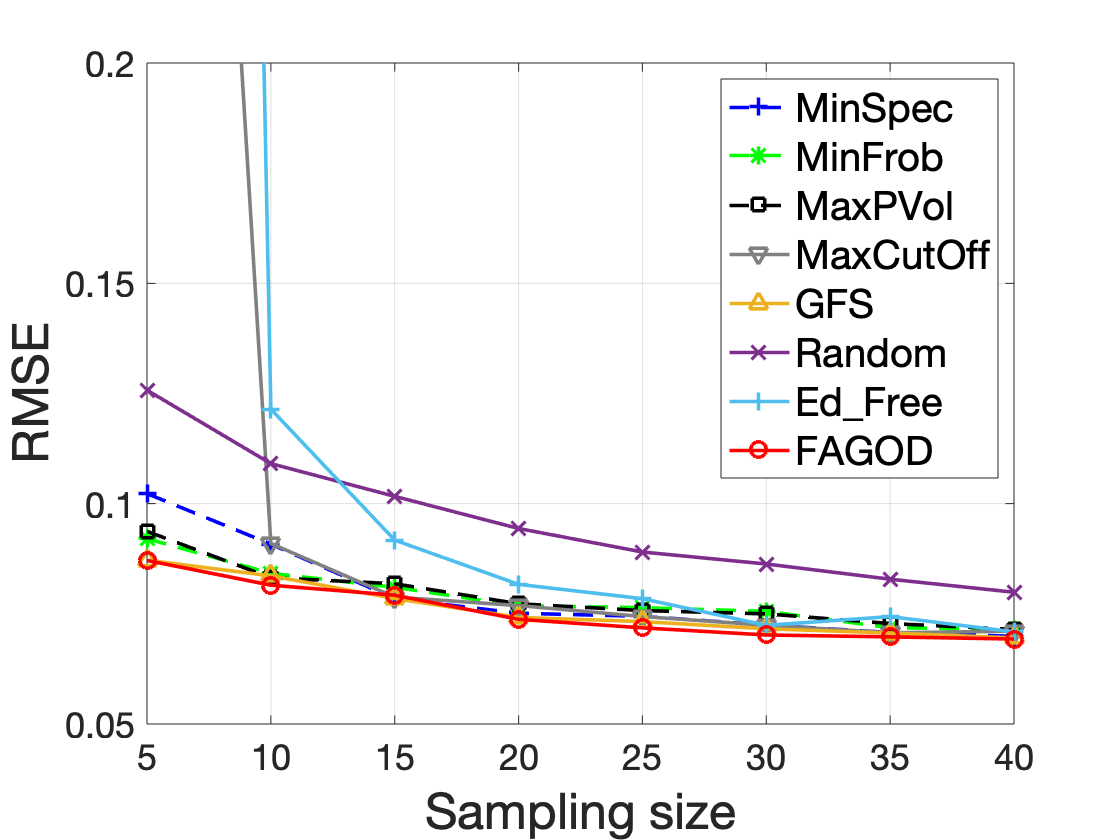}
		\end{minipage}
	}
	\centering
	\caption{Reconstruction results for different signal and graph models. (a) Graph \textbf{G1} and signal model \textbf{GS1}, (b) Graph \textbf{G2} signal model \textbf{GS1}, (c) Graph \textbf{G3} signal model \textbf{GS1}, (d) Graph \textbf{G1} signal model \textbf{GS2}, (e) Graph \textbf{G2} signal model \textbf{GS2}, (f) Graph \textbf{G3} signal model \textbf{GS2}.}
	\label{Fig3}
\end{figure*}
\begin{figure*}[h!]
	\centering
	\subfigure[]{
		\begin{minipage}[t]{0.3\linewidth}			\includegraphics[width=1\linewidth]{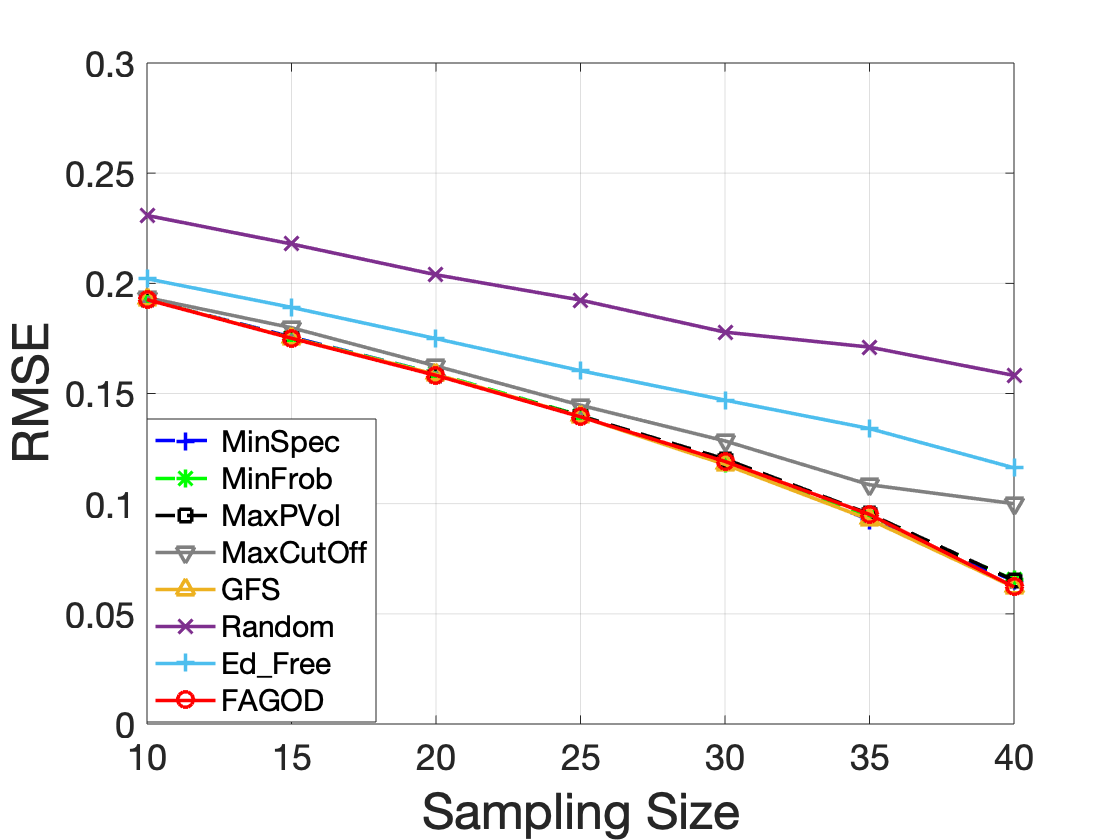}
		\end{minipage}
	}
	\subfigure[]{
		\begin{minipage}[t]{0.3\linewidth}
				\includegraphics[width=1\linewidth]{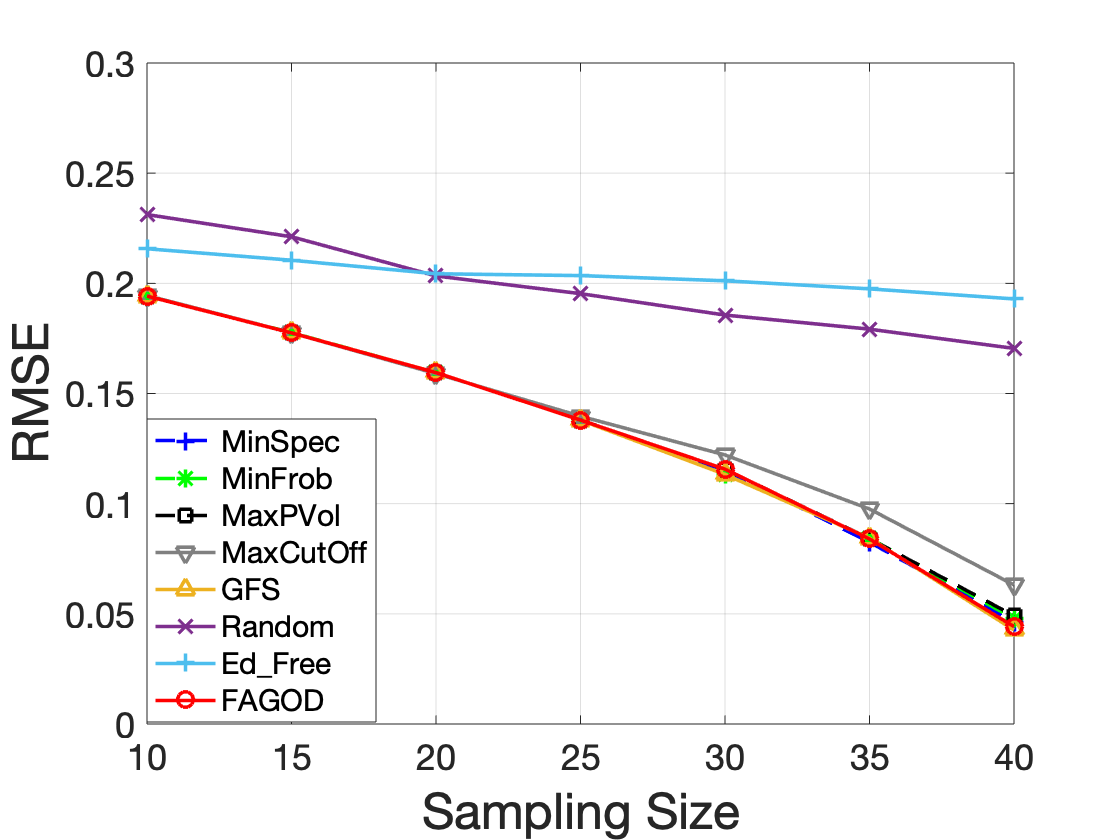}
		\end{minipage}%
	}
	\subfigure[]{
		\begin{minipage}[t]{0.3\linewidth}
			\includegraphics[width=1\linewidth]{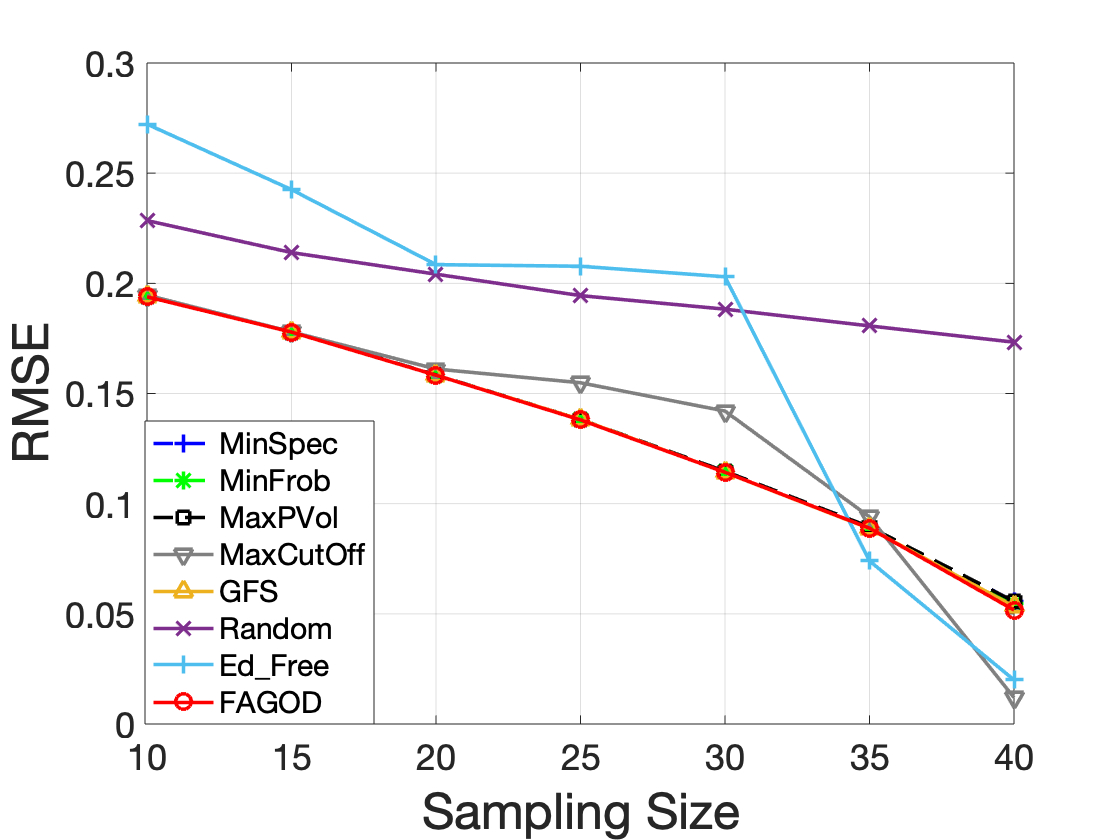}
		\end{minipage}
	}%
	\centering
	\caption{Reconstruction results of signal model \textbf{GS3} for different graph models. (a) Graph \textbf{G1}, (b) Graph \textbf{G2}, (c) Graph \textbf{G3}.}
		\label{Fig5}
\end{figure*}
\begin{figure*}[h!]
	\centering
	\subfigure[]{
		\begin{minipage}[t]{0.3\linewidth}			\includegraphics[width=1\linewidth]{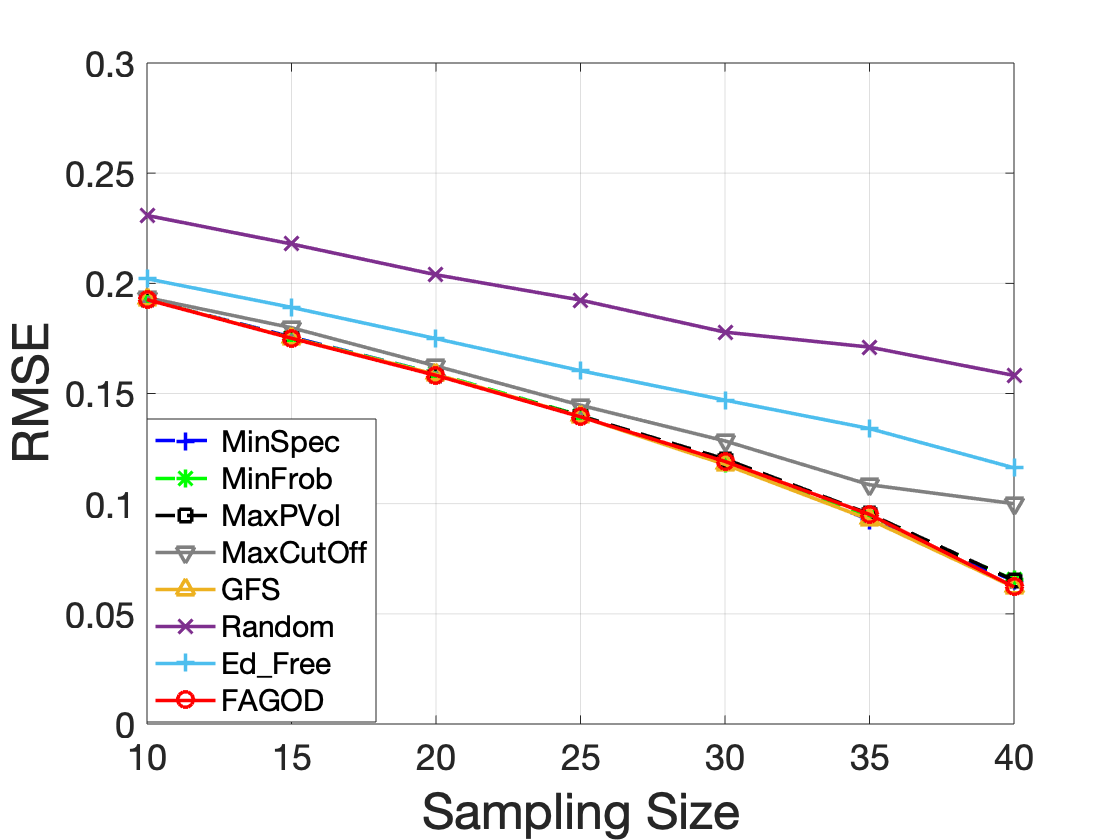}
		\end{minipage}
	}
	\subfigure[]{
		\begin{minipage}[t]{0.3\linewidth}
				\includegraphics[width=1\linewidth]{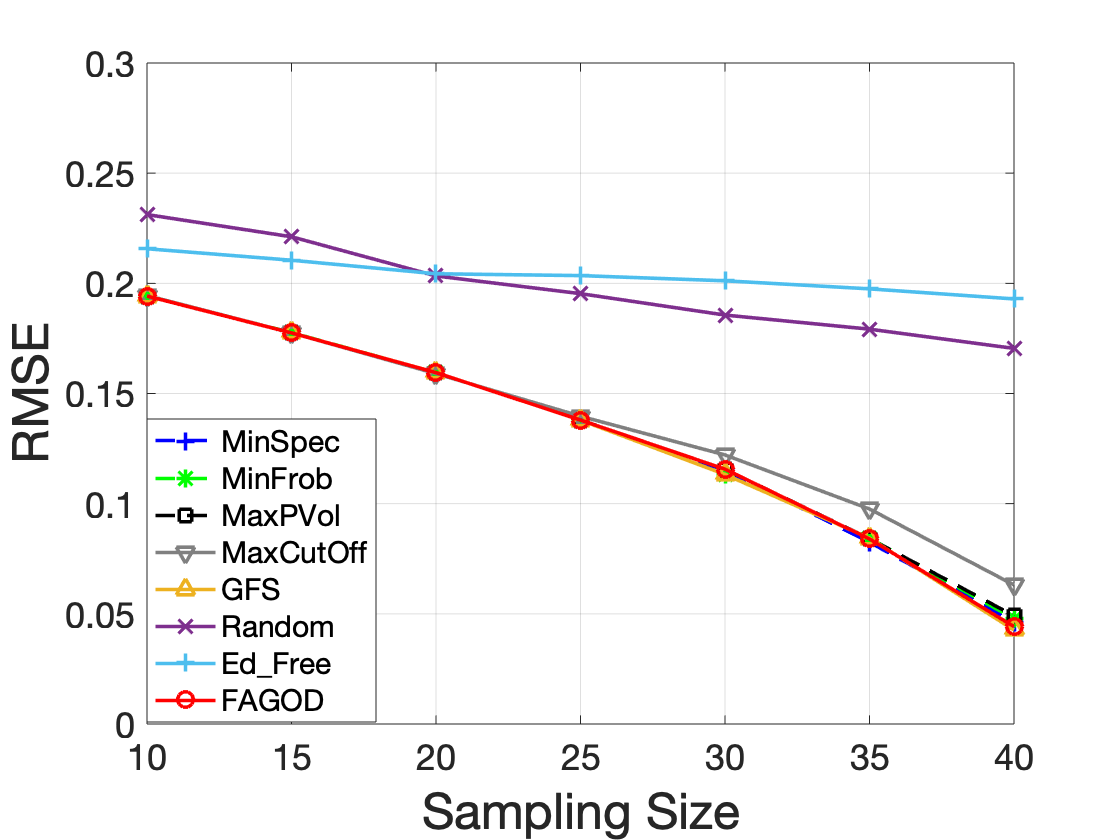}
		\end{minipage}%
	}
	\subfigure[]{
		\begin{minipage}[t]{0.3\linewidth}
			\includegraphics[width=1\linewidth]{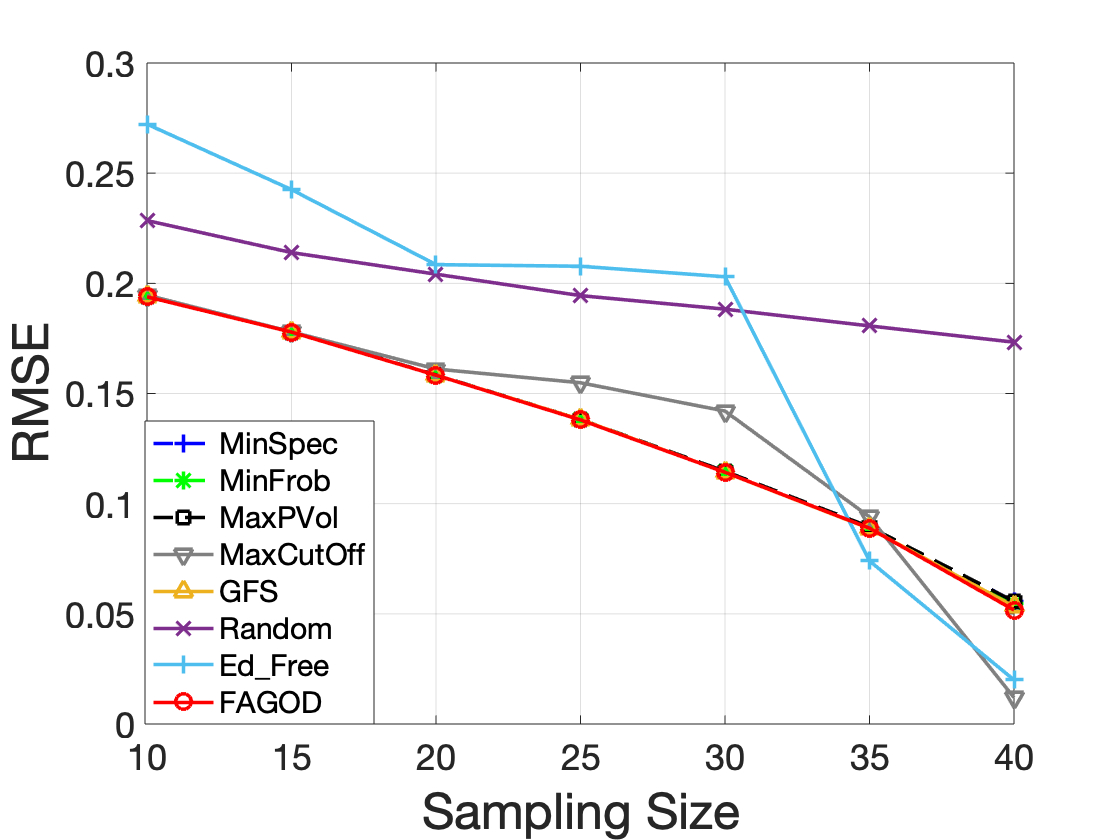}
		\end{minipage}
	}%
	\centering
	\caption{Reconstruction results of signal model \textbf{GS3} for different graph models. (a) Graph \textbf{G1}, (b) Graph \textbf{G2}, (c) Graph \textbf{G3}.}
		\label{Fig5}
\end{figure*}

\begin{figure*}[t!]
	\centering
	\subfigure[]{
		\begin{minipage}[t]{0.3\linewidth}			\includegraphics[width=1\linewidth]{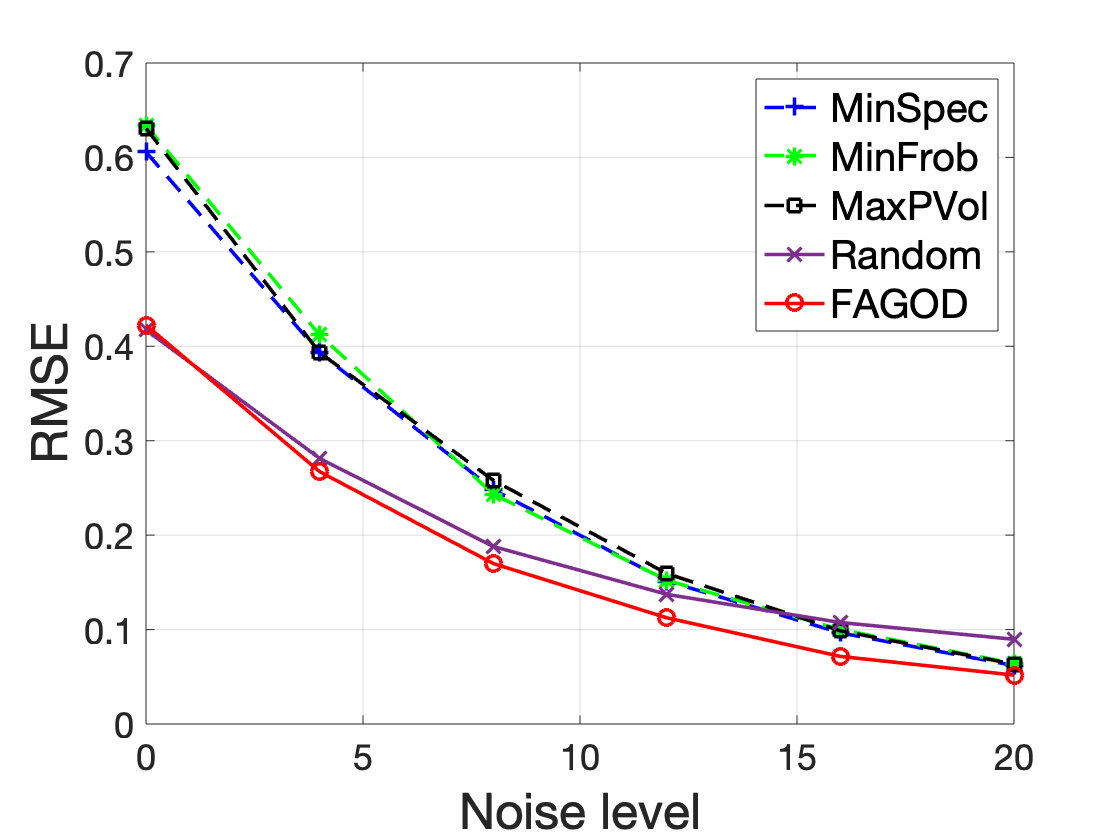}
		\end{minipage}
	}
	\subfigure[]{
		\begin{minipage}[t]{0.3\linewidth}
				\includegraphics[width=1\linewidth]{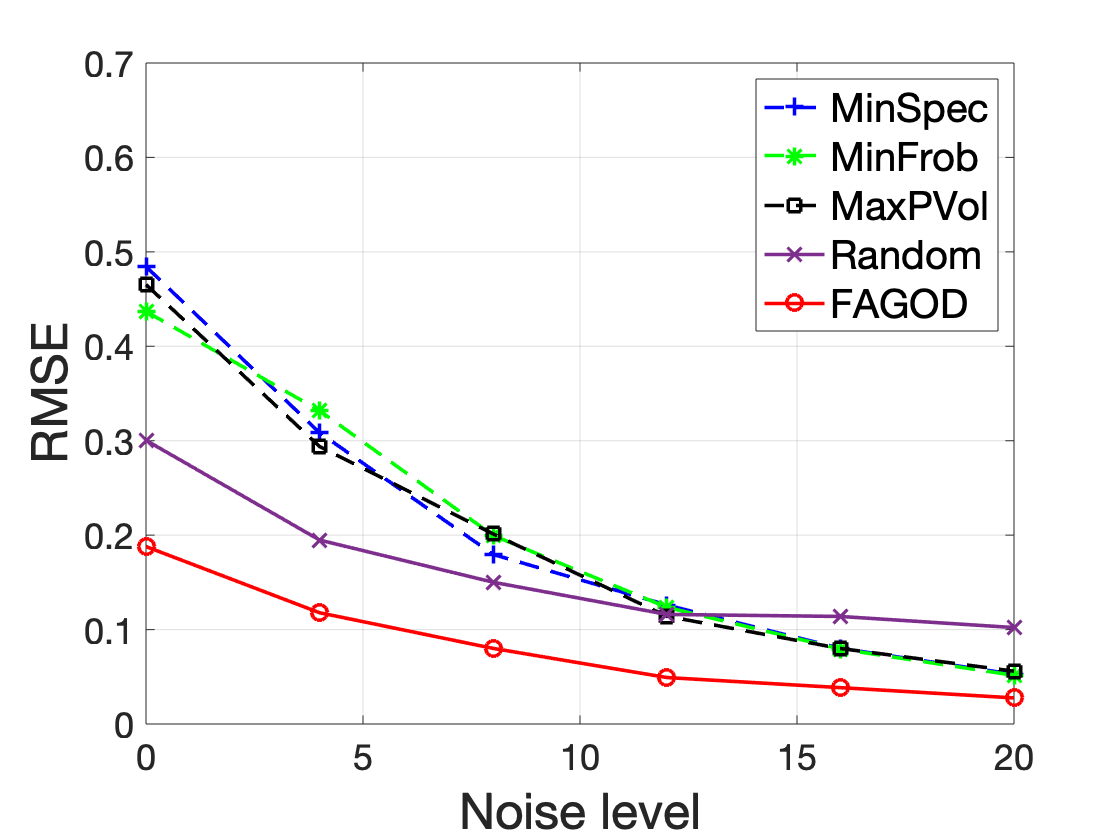}
		\end{minipage}%
	}
	\subfigure[]{
		\begin{minipage}[t]{0.3\linewidth}
			\includegraphics[width=1\linewidth]{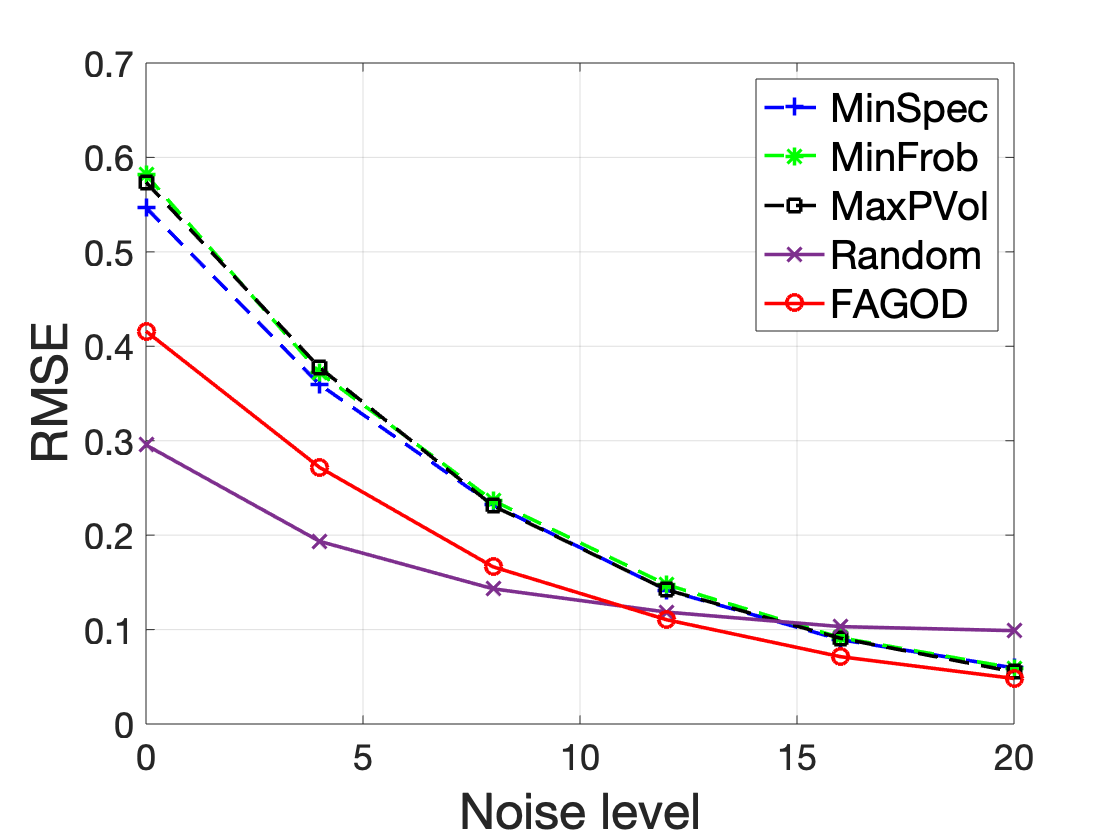}
		\end{minipage}
	}%
	\centering
	\caption{Reconstruction results of signal model \textbf{GS1} for different noise level and graph models. (a) Graph \textbf{G1}, (b) Graph \textbf{G2}, (c) Graph \textbf{G3}.}
	\label{Fig4}
\end{figure*}

\begin{figure*}[t!]
	\centering
	\subfigure[]{
		\begin{minipage}[t]{0.3\linewidth}			\includegraphics[width=1\linewidth]{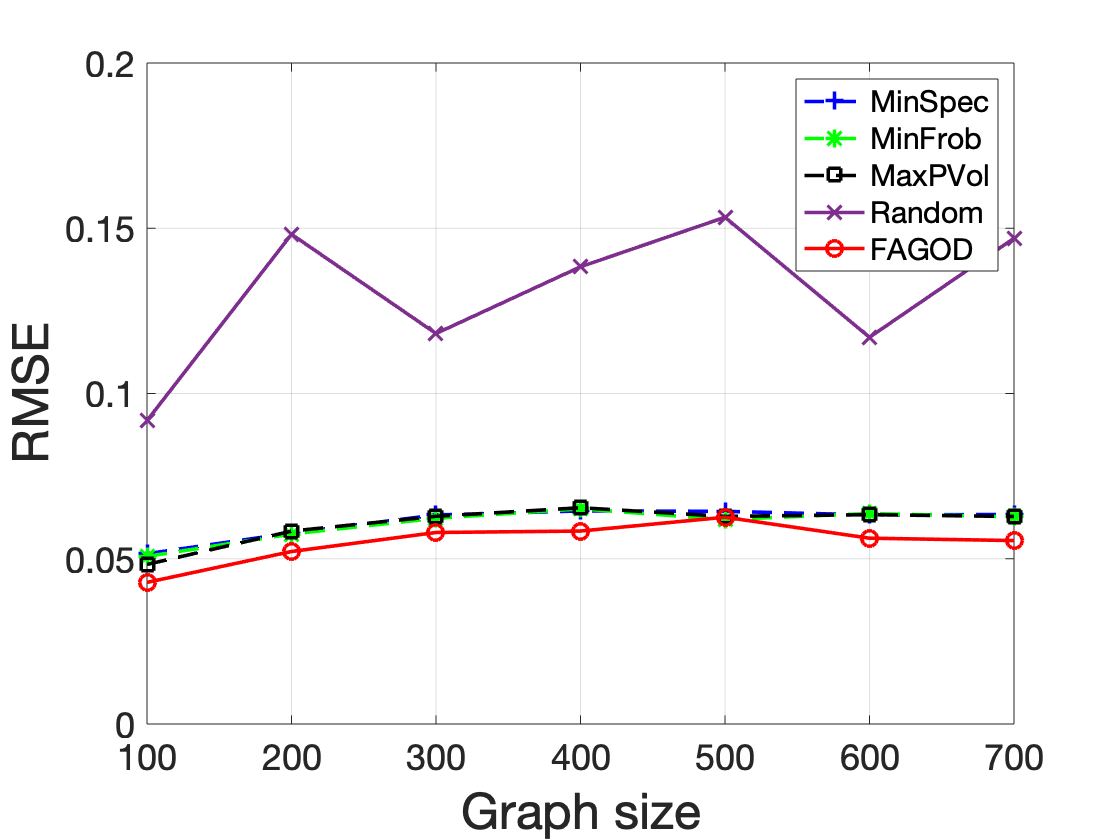}
		\end{minipage}
	}
	\subfigure[]{
		\begin{minipage}[t]{0.3\linewidth}
				\includegraphics[width=1\linewidth]{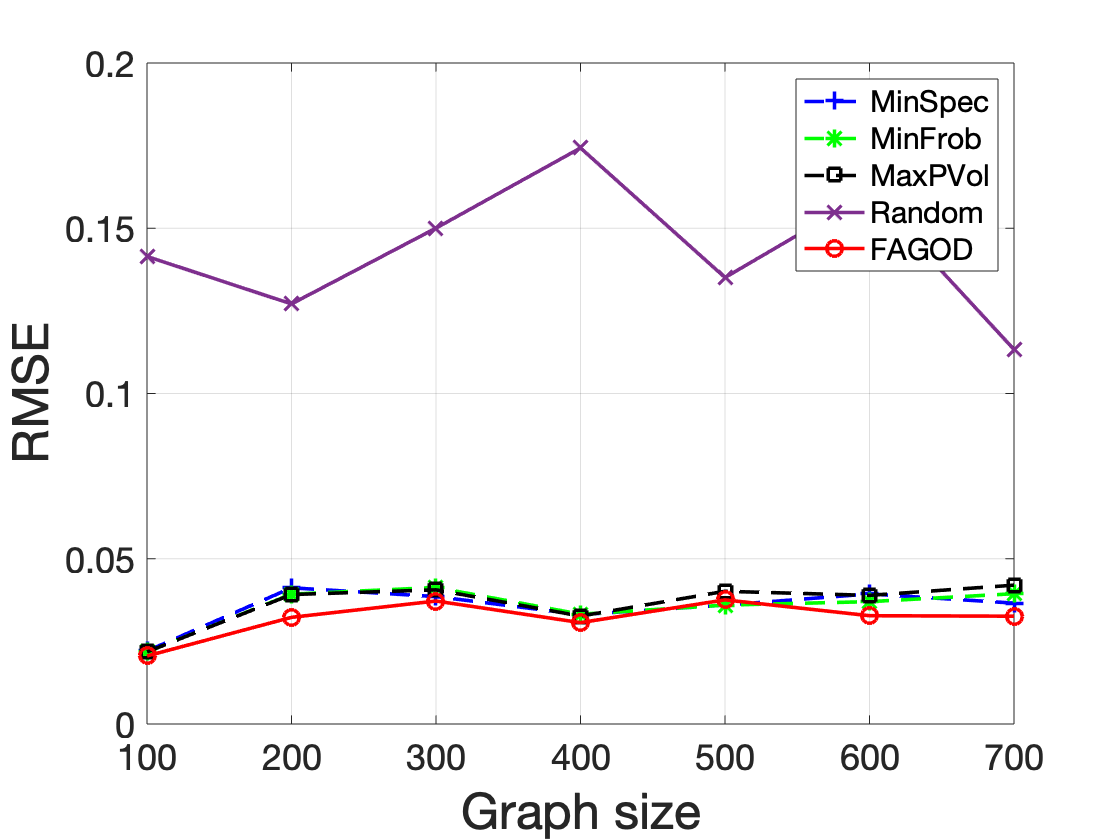}
		\end{minipage}%
	}
	\subfigure[]{
		\begin{minipage}[t]{0.3\linewidth}
			\includegraphics[width=1\linewidth]{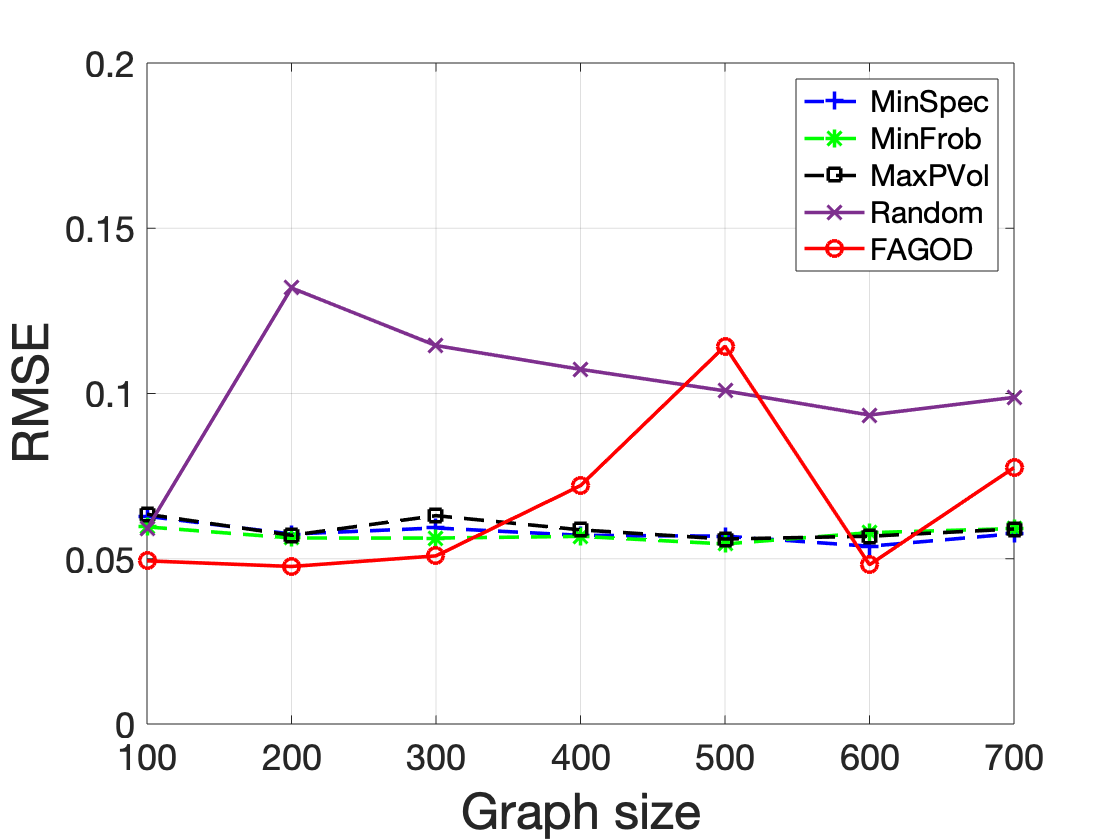}
		\end{minipage}
	}%
	\centering
	\caption{Reconstruction results of signal model \textbf{GS1} for different graph size. (a) Graph \textbf{G1}, (b) Graph \textbf{G2}, (c) Graph \textbf{G3}.}
	\label{Fig6}
\end{figure*}

The average RMSE in terms of sample size is shown in Fig.~\ref{Fig3}. As shown, our proposed FAGOD achieves comparable or smaller RMSE values to the competing deterministic schemes on all three graphs and achieves much better performance than the random sampling scheme Random. Especially when the number of samples is equal to the bandwidth, our method performs the best among the compared methods. This indicates that our proposed algorithm is especially effective with a small sampling size. Although our method is designed for strictly bandlimited graph signals, it performs well for approximated graph signals.

	For eigendecomposition-based algorithms such as MinSpec, MinFrob, MaxPVol, we notice that the performances of these algorithms are satisfactory. However, these eigen-decomposition-based algorithms are computationally too expensive for large graphs. For eigen-decomposition-free algorithms, MaxCutOff and Ed\_Free work well for bandlimited graph signals, but their performance is poor for approximate bandlimited graph signals with small sampling sizes. It has been observed that, in general, the performance of \emph{Random} is not comparable to deterministic graph sampling algorithms.

	We further test all the sampling algorithms with the sampling size to be no smaller than the bandwidth for \textbf{GS3}, since most sampling methods are under the assumption that $|\mathcal{S}| \geq K$. As shown in Fig.~\ref{Fig5}, our method achieves similar performance, as eigen-decomposition-based methods.

\subsection{RMSE vs. Noise Level}

We test the performance of these algorithms for bandlimited signal \textbf{GS1} under different observation noise. We set $K = 10$ and sampling size $10$. All noises are generated from Gaussian distribution with mean $0$ and different variances $\sigma^2$. We range the signal-to-noise ratio (SNR) from $0$ to $20$, which is defined as
\begin{equation}
	\text{SNR} = 10\log\left(\frac{5\times10^{-1}}{\sigma^2}\right).
\end{equation}
Sampling size equals bandwidth for all cases. Fig.~\ref{Fig4} shows that FAGOD performs best even with a small SNR. We also observe that when the noise level is large, Random performs better than eigendecomposition-based algorithms. Its performance is guaranteed due to the restricted isometry property of the measurement matrix proved in~\cite{Randomsampling}.

\subsection{RMSE vs. Graph Size}

We set $K = N/20$ and let sampling size equals bandwidth $K$. For each graph size, we generate one corresponding graph model and $150$ graph signals. Fig.~\ref{Fig6} shows the results for different graph sizes on three graphs. It can be found that our method almost always achieves the best reconstruction effect.

\section{Conclusion}\label{sec7}
In this paper, we have proposed an efficient subset selection method based on the optimal experimental design. Our proposed method combines sampling with G-optimal, and does not require any eigenvalue decomposition of the graph Laplacian matrix. We have also  analyzed for the submodular property of the objective function and provided an estimate on the lower bound of the parameter $\alpha$. Through experiments, we have demonstrated that the reconstruction error of our method has competitive performance compared to the state-of-the-art ones in various scenarios.

\section*{Appendix}
\section{Proof of Proposition 1}\label{prolemma1}
Proposition 1 is obvious considering the following lemma.
\begin{lemma}\label{lemma1}
	Suppose $\mathbf{C}$ is a $k\times k $ positive definite matrix, with $\mathbf{\delta}(\mathbf{C}) = [c_{11},\cdots,c_{kk}]$ of diagonal elements and vector $\boldsymbol{\lambda}(\mathbf{C}) = [\lambda_1,\cdots,\lambda_k]$ of eigenvalues, we have
	\begin{equation}
		\underset{i}{\max}~c_{ii} \geq \left(\prod_{i=1}^kc_{ii}\right)^{\frac{1}{k}} \geq \left(\prod_{i=1}^k\lambda_i\right)^{\frac{1}{k}}.
	\end{equation}
\end{lemma}
\begin{proof}[Proof of Lemma~\ref{lemma1}]
	We write $\mathbf{c} = \delta(\mathbf{C}), \boldsymbol{\lambda} = \lambda(\mathbf{C})$, for short. We choose an eigenvalue decomposition $\mathbf{C} = \mathbf{U}^\top \mathbf{\Lambda} \mathbf{U}=\sum_i \lambda_i \mathbf{u}_i \mathbf{u}_i^\top$, where $\mathbf{u}_i \in \mathbb{R}^{k\times1}$ is the eigenvector. Define a matrix $\mathbf{T} \in \mathbb{R}^{k \times k}$ with entry $t_{ij} = u_{ji}^2 \geq 0$. We have
	\begin{equation}
	c_{i i}=\mathbf{e}_{i}^{\top} \left(\sum_{j = 1}^k \lambda_{j} \mathbf{u}_j \mathbf{u}_j^\top\right) \mathbf{e}_{i}=\sum_{j = 1}^k u_{j i}^{2} \lambda_{j}=\sum_{j = 1}^k t_{i j} \lambda_{j}.
	\end{equation}
	Notice $g(x) = \log x$ is a concave function, so
	\begin{equation}
	\log c_{ii} = \log\left(\sum_{j = 1}^k t_{i j} \lambda_{j}\right) \geq \sum_{j = 1}^k t_{i j}\log \lambda_{j}.
	\end{equation}
	Summation over $i$ gives
	\begin{equation}
	\sum_{i=1}^k \log c_{ii} \geq \sum_{i=1}^k\sum_{j = 1}^k t_{i j}\log \lambda_{j} \overset{(\text{a})}{=} \sum_{j = 1}^k \log \lambda_i.
	\end{equation}
	Here (a) is because the orthogonality of $\mathbf{U}$. Using the property of $\log$ function, we have
	\begin{equation}
		\underset{i}{\max}~c_{ii} \geq \left(\prod_{i=1}^kc_{ii}\right)^{\frac{1}{k}} \geq \left(\prod_{i=1}^k\lambda_i\right)^{\frac{1}{k}}.
	\end{equation}
	Proof finished.
\end{proof}
\begin{proof}[Proof of Proposition 1]
	Lemma~\ref{lemma1} shows that if we attempt to minimize the left side of the inequality, we actually minimize the upper bound of the right side. It is obvious that $\mathbf{V}_{\mathcal{S} K}^{\top} \mathbf{V}_{\mathcal{S} K}$ is semi-positive definite for any sampling set $\mathcal{S}\subset\mathcal{V}$ and $|\mathcal{S}|\geq K$. Thus $\mathbf{V}_{\mathcal{S} K}^{\top} \mathbf{V}_{\mathcal{S} K}+\mu\mathbf{I}$ is always positive definite for $\forall \mu >0$, which indicates that matrix $(\mathbf{V}_{\mathcal{S} K}^{\top} \mathbf{V}_{\mathcal{S} K}+\mu\mathbf{I})^{-1}$ satisfies the conditions in Proposition~\ref{proposition1}. That is to say, minimizing the objective function~\eqref{11} of GOD is equivalent to minimizing an upper bound of the objective function of D-optimal.
\end{proof}

\section{Proof of Theorem~\ref{mainTh}}\label{protheorem}
We first prove that the following lemma holds:
\begin{lemma}\label{beforetheo}
For $\forall$ $\mathbf{A}, \mathbf{B}\in \mathbb{R}^{N\times N}$, we have
	\begin{equation}\label{37}
		-\operatorname{d}(\mathbf{B}-\mathbf{A}) \leq \operatorname{d}(\mathbf{A}) - \operatorname{d}(\mathbf{B}) \leq \operatorname{d}(\mathbf{A}- \mathbf{B}).
	\end{equation}
where $\operatorname{d}(\mathbf{A}) = \max\operatorname{diag}(\mathbf{A})$ means calculating the maximum diagonal element of of the matrix $\mathbf{A}$. 
\end{lemma}
\begin{proof}[Proof of Lemma~\ref{beforetheo}]
	Instead of proving~\eqref{37}, we prove another description of Lemma~\ref{beforetheo}:\\
	
	\emph{For any two sequences $\mathbf{p} = [p_1, p_2, \dots, p_n]$ and $\mathbf{q} = [q_1, q_2, \dots, q_n]$, the following inequality holds
	\begin{equation}\label{500}
	 -\max(\mathbf{q}-\mathbf{p})	\leq\max(\mathbf{p})-\max(\mathbf{q})  \leq \max(\mathbf{p}-\mathbf{q}),
	\end{equation}
	where $\max()$ denotes finding the maximum element in the sequence.}\\
	
	Suppose $p_s = \max(\mathbf{p}), q_t = \max(\mathbf{q})$. First, it is obvious that
	\begin{equation}
		\begin{split}
			\max(\mathbf{p}-\mathbf{q}) = p_s-q_i \text{ or } \max(\mathbf{p}-\mathbf{q}) = p_j-q_t,
		\end{split}
	\end{equation}	
	where $1\leq i,j\leq n , i\neq t, j\neq s.$ Without loss of generality, suppose
	\begin{equation}
		\max(\mathbf{p}-\mathbf{q}) = p_s-q_i.
	\end{equation}
	Thus,
	\begin{equation}\label{53}
		\max(\mathbf{p}-\mathbf{q}) = p_s-q_i \geq p_s-q_t = \max(\mathbf{p})-\max(\mathbf{q}).
	\end{equation}
	According to~\eqref{53}, 
		\begin{align}
			-\max(\mathbf{q}-\mathbf{p}) &=\max(\mathbf{q})-\max(\mathbf{q}-\mathbf{p})-\max(\mathbf{q})\nonumber\\
			&\leq \max(\mathbf{q}-\mathbf{q}+\mathbf{p})-\max(\mathbf{q})\nonumber\\
			&=\max(\mathbf{p})-\max(\mathbf{q}).
		\end{align}
	Proof finished.
\end{proof}

Based on Lemma~\ref{beforetheo}, we now give the proof of Theorem~\ref{mainTh}.
\begin{proof}[Proof of Theorem~\ref{mainTh}]
	We define \(Z(\mathcal{S})  = \sum_{i \in \mathcal{S}} \mathbf{v}_{i:}^\top \mathbf{v}_{i:}  +  \mu \mathbf{I}\) and  \(Z(\mathcal{S}\cup\{j\})  = Z(\mathcal{S}) + \mathbf{v}_{j:}^\top \mathbf{v}_{j:}\). The objective function~\eqref{21} can be written as
		\begin{align}
			\mathcal{S}&=\underset{|\mathcal{S}|=M}{\arg \min }~f\left(\left(Z(\mathcal{S})\right)^{-1}\right)\nonumber\\
				& = \underset{|\mathcal{S}|=M}{\arg \min }~g\left(\mathcal{S}\right),
		\end{align}
	where $f(\mathbf{A}) = \max\operatorname{diag}(\mathbf{A})$ means calculating the maximum diagonal element of of the matrix $\mathbf{A}$. Our proof contains two parts.
	
	\begin{enumerate}
		\item[i.] \emph{Monotone decreasing}\\		
		We first introduce the following lemma, which is known as \emph{Sherman-Morrison formula}~\cite{bach2013learning}:
		\begin{lemma}
			Suppose \(\mathbf{A} \in \mathbb{R}^{N \times N}\) is an invertible square matrix and \(\mathbf{u}, \mathbf{v} \in \mathbb{R}^{N}\) are column vectors. If \(\mathbf{A}+\mathbf{u v}^{\top}\) is invertible,then its inverse is given by
	\begin{equation}
		\left(\mathbf{A}+\mathbf{u v}^{\top}\right)^{-1}=\mathbf{A}^{-1}-\frac{\mathbf{A}^{-1} \mathbf{u v}^{\top} \mathbf{A}^{-1}}{1+\mathbf{v}^{\top} \mathbf{A}^{-1} \mathbf{u}}.
	\end{equation}
		\end{lemma}
		According to our definition, we have
		\begin{align}
		g(\mathcal{S}) &= f \left( Z(\mathcal{S})^{-1}  \right), \nonumber\\
		g(\mathcal{S}\cup\{j\}) & = f\left(Z(\mathcal{S}\cup\{j\})^{-1}\right)\nonumber\\
		&= f \left( \left(Z(\mathcal{S}) + \mathbf{v}_{j:}^\top \mathbf{v}_{j:} \right)^{-1}\right) \nonumber\\
		& = f\left(   Z(\mathcal{S})^{-1} - \frac{Z(\mathcal{S})^{-1} \mathbf{v}_{j:}^\top \mathbf{v}_{j:} Z(\mathcal{S})^{-1} }{1 + \mathbf{v}_{j:} Z(\mathcal{S})^{-1}\mathbf{v}_{j:}^\top} \right)\nonumber\\
		\end{align} 
		Define
		\begin{equation}
			C(\mathcal{S},j) = \frac{Z(\mathcal{S})^{-1} \mathbf{v}_{j:}^\top \mathbf{v}_{j:} Z(\mathcal{S})^{-1} }{1 + \mathbf{v}_{j:} Z(\mathcal{S})^{-1}\mathbf{v}_{j:}^\top}.
		\end{equation}
		We now prove that all diagonal elements of $C(\mathcal{S},j)$ is non-negative. It is obvious that for any $\mathcal{S}, Z(\mathcal{S})$ is positive definite and the eigenvalues are in \([\mu,1+\mu]\). Thus, $Z(\mathcal{S})^{-1}$ is also positive definite and we have 
		\begin{equation}\label{49}
		0 < 1+\frac{1}{1+\mu} \leq 1 + \mathbf{v}_{j:} Z(\mathcal{S})^{-1}\mathbf{v}_{j:}^\top \leq 1+\frac{1}{\mu}.
		\end{equation}
		Also, the $k$-th diagonal element of \(Z(\mathcal{S})^{-1} \mathbf{v}_{j:}^\top \mathbf{v}_{j:} Z(\mathcal{S})^{-1}\) satisfies
			\begin{align}		\label{50}
				&\mathbf{e}_k^\top Z(\mathcal{S})^{-1} \mathbf{v}_{j:}^\top \mathbf{v}_{j:} Z(\mathcal{S})^{-1} \mathbf{e}_k \nonumber\\
				=&~ \left(\mathbf{e}_k Z(\mathcal{S})^{-1}\right)^\top \mathbf{v}_{j:} ^\top\mathbf{v}_{j:} \left(Z(\mathcal{S})^{-1} \mathbf{e}_k\right) 
				 \geq 0,
			\end{align}
		where \(\mathbf{e}_k\) denotes the $n$-dimensional unit column vector with the $k$-th entry equals 1 and the last inequality holds because $\mathbf{v}_{j:}^\top \mathbf{v}_{j:}$ is semi-positive definite. Combining~\eqref{49} and~\eqref{50}, all diagonal elements of matrix $ C(\mathcal{S},j)$ is no smaller than $0$. So according to Lemma~\ref{beforetheo},
		
		\begin{align}
			&f \left( Z(\mathcal{S})^{-1}  \right) -f\left(   Z(\mathcal{S})^{-1} - \frac{Z(\mathcal{S})^{-1} \mathbf{v}_{j:}^\top \mathbf{v}_{j:} Z(\mathcal{S})^{-1} }{1 + \mathbf{v}_{j:} Z(\mathcal{S})^{-1}\mathbf{v}_{j:}^\top} \right) \nonumber\\
			\geq&~ - f\left(   - \frac{Z(\mathcal{S})^{-1} \mathbf{v}_{j:}^\top \mathbf{v}_{j:} Z(\mathcal{S})^{-1} }{1 + \mathbf{v}_{j:} Z(\mathcal{S})^{-1}\mathbf{v}_{j:}^\top} \right) \nonumber\\
			=&~ f\left( C(\mathcal{S},j)\right)  \geq 0,
		\end{align}
		which means
		\begin{equation}
		g(\mathcal{S}\cup\{j\}) - g(\mathcal{S}) \leq 0.
		\end{equation}
		
		This shows that set function $g(\mathcal{S}) = f \left( Z(\mathcal{S})^{-1}  \right)$ is a monotone decreasing function.

		\vskip 0.5em

		\item[ii.] \emph{\(\alpha\)-supermodular}\\ 
		We now prove that our proposed function is \(\alpha\)-supermodular and derive the lower bound of $\alpha$. Recall
		\begin{align}\label{62}
		\alpha &=\min _{\mathcal{A} \subseteq \mathcal{B} \subseteq \mathcal{V}, j \notin \mathcal{B}} \frac{g(\mathcal{A} \cup\{j\})-g(\mathcal{A})}{g(\mathcal{B} \cup\{j\})-g(\mathcal{B})}\nonumber\\
		& = \min _{\mathcal{A} \subseteq \mathcal{B} \subseteq \mathcal{V}, j \notin \mathcal{B}} \frac{f\left(Z(\mathcal{A}\cup\{j\})^{-1}\right) - f\left(Z(\mathcal{A})^{-1}\right)}{f\left(Z(\mathcal{B}\cup\{j\})^{-1}\right) - f\left(Z(\mathcal{B})^{-1}\right)}. 
		\end{align}
		According to Lemma~\ref{beforetheo}, we have
		\begin{align}
		\alpha	& \geq \frac{-f\left( Z(\mathcal{A})^{-1} - Z(\mathcal{A}\cup\{j\})^{-1}\right) }{f(Z\left(\mathcal{B}\cup\{j\})^{-1} - Z(\mathcal{B})^{-1}\right	)}\nonumber \\
		& = \frac{-f(C(\mathcal{A},j))}{-f(C(\mathcal{B},j))} \nonumber\\
		& = -\frac{1 + \mathbf{v}_{j:} Z(\mathcal{B})^{-1}\mathbf{v}_{j:}^\top}{1 + \mathbf{v}_{j:} Z(\mathcal{A})^{-1}\mathbf{v}_{j:}^\top} \frac{f \left(Z(\mathcal{A})^{-1} \mathbf{v}_{j:}^\top \mathbf{v}_{j:} Z(\mathcal{A}\right)^{-1})}{f \left(-Z(\mathcal{B})^{-1} \mathbf{v}_{j:} ^\top\mathbf{v}_{j:} Z(\mathcal{B})^{-1}\right)}.
		\end{align}
		
		For simplicity, let 
			\begin{align}
				t_{\max}(\mathcal{S},j) &= f \left(Z(\mathcal{S})^{-1} \mathbf{v}_{j:} ^\top\mathbf{v}_{j:} Z(\mathcal{S}\right)^{-1})\nonumber\\
				&= \max \operatorname{diag}(Z(\mathcal{S})^{-1}\mathbf{v}_{j:} ^\top\mathbf{v}_{j:} Z(\mathcal{S})^{-1}),\nonumber\\
				t_{\min}(\mathcal{S},j) &= \min \operatorname{diag}(Z(\mathcal{S})^{-1} \mathbf{v}_{j:} ^\top\mathbf{v}_{j:} Z(\mathcal{S})^{-1}).
			\end{align}
		Here $\min\operatorname{diag}$ means the minimum element of the diagonal elements of the matrix. It is obvious
		\begin{equation}
			t_{\max}(\mathcal{S},j) \geq t_{\min}(\mathcal{S},j).
		\end{equation}
		Now we arrive at 
		\begin{equation}
			\alpha \geq -\frac{1 + \mathbf{v}_{j:} Z(\mathcal{B})^{-1}\mathbf{v}_{j:}^\top}{1 + \mathbf{v}_{j:} Z(\mathcal{A})^{-1}\mathbf{v}_{j:}^\top}    \frac{t_{\max}(\mathcal{A},j)}{t_{\min}(\mathcal{B},j)}.
		\end{equation}
		
		We first prove that $t_{\max}(\mathcal{S},j)$ is a monotone decreasing function for $\mathcal{S}$. Recall that 
		\begin{align}
		&t_{\max}(\mathcal{S}\cup\{j\})\nonumber \\
		=&~ f \left(Z(\mathcal{S}\cup\{j\})^{-1} \mathbf{v}_{j:} ^\top\mathbf{v}_{j:}Z(\mathcal{S}\cup\{j\})^{-1}\right)\nonumber\\
		=&~ f \left((Z(\mathcal{S})+ \mathbf{v}_{j:} ^\top\mathbf{v}_{j:})^{-1} \mathbf{v}_{j:} ^\top\mathbf{v}_{j:}(Z(\mathcal{S}) + \mathbf{v}_{j:} ^\top\mathbf{v}_{j:})^{-1}\right)\nonumber\\
		=&~t_{\max}(\mathcal{S}) - f(h(\mathcal{S},j)),
		\end{align}
		where
		\begin{align}
		h(\mathcal{S},j) &\triangleq \frac{2\mathbf{v}_{j:} \left(Z(\mathcal{S})^{-1}\mathbf{v}_{j:}^\top\right)^2 \mathbf{v}_{j:} Z(\mathcal{S})^{-1}}{1 + \mathbf{v}_{j:} Z(\mathcal{S})^{-1}\mathbf{v}_{j:}^\top} \nonumber\\
		&-\frac{\left(\mathbf{v}_{j:} Z(\mathcal{S})^{-1}\mathbf{v}_{j:}^\top\right)^2Z(\mathcal{S})^{-1} \mathbf{v}_{j:} ^\top\mathbf{v}_{j:} Z(\mathcal{S})^{-1}}{1+\mathbf{v}_{j:} Z(\mathcal{S})^{-1}\mathbf{v}_{j:}^\top}.
		\end{align}
		
		We only need to focus on the symbol of the diagonal elements of $h(\mathcal{S},j)$. Notice that
		\begin{align}
		&h(\mathcal{S},j)\nonumber \\
		=&~ \left(2 - \frac{\mathbf{v}_{j:} Z(\mathcal{S})^{-1}\mathbf{v}_{j:}^\top}{1+\mathbf{v}_{j:} Z(\mathcal{S})^{-1}\mathbf{v}_{j:}^\top} \right)Z(\mathcal{S})^{-1}\mathbf{v}_{j:} ^\top\mathbf{v}_{j:}Z(\mathcal{S})^{-1}.
		\end{align}
		
		As shown in~\eqref{49} and~\eqref{50}, all diagonal elements of $Z(\mathcal{S})^{-1}\mathbf{v}_{j:} ^\top\mathbf{v}_{j:}Z(\mathcal{S})^{-1}$ are not smaller than $0$, and
		\begin{equation}
			2 - \frac{C(\mathcal{S},j)}{1+C(\mathcal{S},j)} > 1.
		\end{equation}
		This indicates that the diagonal elements of $h(\mathcal{S},j)$ are all no smaller than $0$. Thus
		\begin{align}
			t_{\max}(\mathcal{S}\cup\{j\})\nonumber &= t_{\max}(\mathcal{S}) - f(h(\mathcal{S},j))\\
			&\leq t_{\max}(\mathcal{S}).
		\end{align}
		which means $t_{\max}(\mathcal{S},j)$ is a monotone decreasing function. Because  $\mathcal{A}\subseteq \mathcal{B}$, 
		\begin{align}
		\alpha &\geq -\frac{1 + \mathbf{v}_{j:} Z(\mathcal{B})^{-1}\mathbf{v}_{j:}^\top}{1 + \mathbf{v}_{j:} Z(\mathcal{A})^{-1}\mathbf{v}_{j:}^\top}    \frac{t_{\max}(\mathcal{A},j)}{t_{\min}(\mathcal{B},j)}\nonumber \\
		&\geq -\frac{1 + \mathbf{v}_{j:} Z(\mathcal{B})^{-1}\mathbf{v}_{j:}^\top}{1 + \mathbf{v}_{j:} Z(\mathcal{A})^{-1}\mathbf{v}_{j:}^\top}\frac{t_{\max}(\mathcal{B},j)}{t_{\min}(\mathcal{B},j)}\nonumber\\
		&\geq -\frac{1 + \mathbf{v}_{j:} Z(\mathcal{B})^{-1}\mathbf{v}_{j:}^\top}{1 + \mathbf{v}_{j:} Z(\mathcal{A})^{-1}\mathbf{v}_{j:}^\top}\nonumber \\
		&\geq \frac{\mu(2+\mu)}{(1+\mu)^2}.
		\end{align}

	\end{enumerate}
\end{proof}

\section{Proof of Proposition~\ref{proposition2}}\label{pro2}
We use the following lemma~\cite{horn2012matrix}:
\begin{lemma}\label{lemma5}
	Suppose $\mathbf{C}$ is a $m\times n$ matrix, $m \geq n$, then the $n$ eigenvalues of $\mathbf{C}^\top \mathbf{C}$ is also the eigenvalues of $\mathbf{C}\mathbf{C}^\top$. Furthermore, the other $n-m$ eigenvalues of $\mathbf{C}\mathbf{C}^\top$ is $0$.
\end{lemma}

\begin{proof}[Proof of Proposition~\ref{proposition2}]
	Suppose the diagonal elements of $\left( {\mathbf{V}}_{\mathcal{S} K}{\mathbf{V}}_{\mathcal{S} K}^{\top} + \mu \mathbf{I}\right)^{-1}$ is $\{d_{i}\}_{i=1,2,\cdots,M}$. Also suppose the eigenvalues of $\left( {\mathbf{V}}_{\mathcal{S} K}^{\top}{\mathbf{V}}_{\mathcal{S} K} + \mu \mathbf{I}\right)^{-1}$ are $\{\lambda_i\}_{i =1,2,\cdots,K}$. According to Lemma~\ref{lemma5}, the $M$ eigenvalues of $\left( {\mathbf{V}}_{\mathcal{S} K}{\mathbf{V}}_{\mathcal{S} K}^{\top} + \mu \mathbf{I}\right)^{-1}$ are \\$\{\lambda_i, 1/\mu,\cdots,1/\mu\}_{i=1,\cdots,K}$. According to Lemma~\ref{lemma1}, we have
	\begin{equation}
		\underset{i}{\max}~d_{ii} \geq \left(\prod_{i=1}^Md_{ii}\right)^{\frac{1}{M}} \geq \left(\prod_{i=1}^K\lambda_i\right)^{\frac{1}{K}}\left(\frac{1}{\mu}\right)^{\frac{1}{M-K}}.
	\end{equation}
\end{proof}


%

\end{document}